\documentclass{llncs}

\usepackage[utf8]{inputenc}
\usepackage{etex}
\usepackage{microtype}
\usepackage{hyperref}
\usepackage{booktabs}
\usepackage{mathtools}
\usepackage{tabu}
\usepackage{float}
\usepackage{framed} 
\usepackage{rotating}

\usepackage{tikz}
\usetikzlibrary{shadows}
\usetikzlibrary{calc}
\usetikzlibrary{positioning}
\usetikzlibrary{circuits}
\usetikzlibrary{circuits.logic.US}

\newif\ifnopromise
\nopromisetrue

\ifnopromise
\makeatletter
\renewcommand\subsubsection{\@startsection{subsubsection}{3}{\z@}%
                       {-4.5\p@ \@plus -1\p@ \@minus -1\p@}%
                       {-0.25em \@plus -0.11em \@minus -0.05em}%
                       {\normalfont\normalsize\bfseries\boldmath}}
\makeatother
\fi

\makeatletter
\pgfdeclareshape{myplus}{
  \inheritsavedanchors[from=rectangle] 
  \inheritanchorborder[from=rectangle]
  \inheritanchor[from=rectangle]{center}
  \inheritanchor[from=rectangle]{north}
  \inheritanchor[from=rectangle]{south}
  \inheritanchor[from=rectangle]{west}
  \inheritanchor[from=rectangle]{east}
  \backgroundpath{
    \southwest \pgf@xa=\pgf@x \pgf@ya=\pgf@y
    \northeast \pgf@xb=\pgf@x \pgf@yb=\pgf@y
    \pgf@xc=0.5\pgf@xa
    \pgf@yc=0.5\pgf@ya
    \advance\pgf@xc by 0.5\pgf@xb
    \advance\pgf@yc by 0.5\pgf@yb
    \pgfpathmoveto{\pgfpoint{\pgf@xa}{\pgf@ya}}
    \pgfpathlineto{\pgfpoint{\pgf@xa}{\pgf@yb}}
    \pgfpathlineto{\pgfpoint{\pgf@xb}{\pgf@yb}}
    \pgfpathlineto{\pgfpoint{\pgf@xb}{\pgf@ya}}
    \pgfpathclose
    \pgfpathmoveto{\pgfpoint{\pgf@xc}{\pgf@ya}}
    \pgfpathlineto{\pgfpoint{\pgf@xc}{\pgf@yb}}
    \pgfpathclose
    \pgfpathmoveto{\pgfpoint{\pgf@xa}{\pgf@yc}}
    \pgfpathlineto{\pgfpoint{\pgf@xb}{\pgf@yc}}
    \pgfpathclose
 }
}

\pgfdeclareshape{myxor}{
  \inheritsavedanchors[from=rectangle] 
  \inheritanchorborder[from=rectangle]
  \inheritanchor[from=rectangle]{center}
  \inheritanchor[from=rectangle]{north}
  \inheritanchor[from=rectangle]{south}
  \inheritanchor[from=rectangle]{west}
  \inheritanchor[from=rectangle]{east}
  \backgroundpath{
    \southwest \pgf@xa=\pgf@x \pgf@ya=\pgf@y
    \northeast \pgf@xb=\pgf@x \pgf@yb=\pgf@y
    \pgf@xc=0.5\pgf@xa
    \pgf@yc=0.5\pgf@ya
    \advance\pgf@xc by 0.5\pgf@xb
    \advance\pgf@yc by 0.5\pgf@yb
    \pgfutil@tempdima=0.5\pgf@xb
    \advance\pgfutil@tempdima by -0.5\pgf@xa
    \pgfpathmoveto{\pgfpoint{\pgf@xc}{\pgf@ya}}
    \pgfpathlineto{\pgfpoint{\pgf@xc}{\pgf@yb}}
    \pgfpathclose
    \pgfpathmoveto{\pgfpoint{\pgf@xa}{\pgf@yc}}
    \pgfpathlineto{\pgfpoint{\pgf@xb}{\pgf@yc}}
    \pgfpathclose
    \pgfpathcircle{\pgfpoint{\pgf@xc}{\pgf@yc}}{\pgfutil@tempdima}
    \pgfpathclose
 }
}

\pgfdeclareshape{myrot}{
  \inheritsavedanchors[from=rectangle] 
  \inheritanchorborder[from=rectangle]
  \inheritanchor[from=rectangle]{center}
  \inheritanchor[from=rectangle]{north}
  \inheritanchor[from=rectangle]{south}
  \inheritanchor[from=rectangle]{west}
  \inheritanchor[from=rectangle]{east}
  \backgroundpath{
    \southwest \pgf@xa=\pgf@x \pgf@ya=\pgf@y
    \northeast \pgf@xb=\pgf@x \pgf@yb=\pgf@y
    \pgf@xc=0.5\pgf@xa
    \pgf@yc=0.5\pgf@ya
    \advance\pgf@xc by 0.5\pgf@xb
    \advance\pgf@yc by 0.5\pgf@yb
    \pgfutil@tempdima=0.5\pgf@xb
    \advance\pgfutil@tempdima by -0.5\pgf@xa

    \advance\pgf@xa by 0.25\pgfutil@tempdima
    \advance\pgf@ya by 1.707107\pgfutil@tempdima

    \pgfsetarrowsend{to reversed}
    \pgfpathmoveto{\pgfpoint{\pgf@xc}{\pgf@yb}}
    \pgfpatharc{90}{-230}{\pgfutil@tempdima}
    \pgfpathlineto{\pgfpoint{\pgf@xa}{\pgf@ya}}
 }
}

\makeatother

\tikzstyle{fun}=[draw,very thick,fill=white,drop shadow]
\tikzstyle{plus}=[fun,shape=myplus,inner sep=0pt,minimum size=2ex]
\tikzstyle{xor}=[fun,shape=myxor,inner sep=0pt,minimum size=2ex]
\tikzstyle{rot}=[fun,shape=myrot,inner sep=0pt,minimum size=2ex]
\tikzstyle{fork}=[shape=circle,inner sep=0pt,minimum size=3pt,fill=black]

\DeclareMathOperator{\CBCMAC}{CBC-MAC}
\DeclareMathOperator{\PMAC}{PMAC}
\DeclareMathOperator{\GMAC}{GMAC}
\DeclareMathOperator{\OCB}{OCB}
\DeclareMathOperator{\OMD}{OMD}

\let\doublebar\|
\DeclarePairedDelimiter\norm{\doublebar}{\doublebar}

\let\|\relax
\DeclareMathSymbol{\|}{\mathbin}{symbols}{"6B}

 \newcommand{\GL}[1]{}
 \newcommand{\MN}[1]{}
 \newcommand{\AL}[1]{}
 \newcommand{\MK}[1]{}

\newcommand{\ket}[1]{| #1 \rangle}

\newcommand{\hide}[1]{}

\newcounter{subfig}

\begin{document}

\author{Marc Kaplan\inst{1,2} \and Ga\"etan Leurent\inst{3}
Anthony Leverrier\inst{3} \and Mar\'ia Naya-Plasencia\inst{3}}

\institute{LTCI, T\'el\'ecom ParisTech, 23 avenue d'Italie, 75214 Paris CEDEX 13, France
\and
School of Informatics, University of Edinburgh,\\
10 Crichton Street, Edinburgh EH8 9AB, UK
\and
Inria Paris, France}



\title{Breaking Symmetric Cryptosystems using Quantum Period Finding}

\maketitle

\thispagestyle{plain}

\begin{abstract}
Due to Shor's algorithm, quantum computers are a severe threat for
public key cryptography. This motivated the cryptographic community to search for quantum-safe solutions.
On the other hand, the impact of quantum computing on secret key cryptography  is much less understood.
In this paper, we consider attacks where 
an
adversary can query an oracle implementing a cryptographic
primitive in a quantum superposition of different states.  This model gives a
lot of power to the adversary, but recent results show that it is
nonetheless possible to build secure cryptosystems in it.

We study applications of a quantum procedure called \emph{Simon's
  algorithm} (the simplest quantum period finding algorithm) in order to attack symmetric cryptosystems in this model.
Following previous works in this direction, we show that several
classical attacks based on finding collisions
can be dramatically sped up using Simon's algorithm: finding a collision
requires $\Omega(2^{n/2})$ queries in the classical setting, but
when collisions happen with some hidden periodicity, they can be found with only
$O(n)$ queries in the quantum model.

We obtain attacks with very strong implications.  First, we show that
the most widely used modes of operation for authentication and
authenticated encryption (\emph{e.g.}  CBC-MAC, PMAC, GMAC, GCM, and
OCB) are completely broken in this security model.  Our attacks are also
applicable to many CAESAR candidates: CLOC, AEZ, COPA, OTR, POET, OMD,
and Minalpher.  This is quite surprising compared to the situation with
encryption modes: Anand \emph{et al.} show that standard modes are
secure with a quantum-secure PRF.

Second, we show that Simon's algorithm can also be applied to slide
attacks, leading to an exponential speed-up of a classical
symmetric cryptanalysis technique in the quantum model.

\hide{
\begin{enumerate}
\item 
\item First, we show that the most common modes of operation for
  message authentication and authenticated encryption (\emph{e.g.}
  CBC-MAC, GCM, and OCB) are completely broken in this setting.  We
  describe forgery attacks against these standardized modes, and
  against several CAESAR candidates with complexity only $O(n)$, where
  $n$ is the size of the block.
\item The second application is to slide attacks, a popular family of cryptanalysis that is  independent of the number of rounds of the attacked cipher. Our result is the first exponential quantum speed-up of a classical cryptanalysis technique, with complexity dropping from $O(2^{n/2})$ to $O(n)$, where $n$ is the size of the block.
\end{enumerate}
}
\medskip

\end{abstract}

\section{Introduction}
The goal of 
post-quantum cryptography is to prepare cryptographic primitives to resist quantum adversaries, {\em i.e.} adversaries with access to a quantum computer.
Indeed, cryptography would be particularly affected by the development of large-scale quantum computers. While currently used asymmetric cryptographic primitives would suffer from devastating attacks due to Shor's algorithm~\cite{DBLP:journals/siamcomp/Shor97}, 
the status of symmetric ones is not so clear: generic attacks, which define the security of ideal symmetric primitives, would get a quadratic speed-up thanks to Grover's algorithm~\cite{DBLP:conf/stoc/Grover96}, hinting that doubling the key length could restore an equivalent ideal security in the post-quantum world. Even though the community seems to consider the issue settled with this solution~\cite{bernstein2009introduction}, only very little is known about real world attacks, that determine the real security of used primitives. Very recently,  this direction has started to draw attention, and interesting results have been obtained. 
New theoretical frameworks to take into account quantum adversaries have been developed~\cite{boneh11,DBLP:conf/eurocrypt/BonehZ13,DBLP:conf/icits/DamgardFNS13,gagliardoni2015semantic,broadbent2015quantum,alagic2016computational}. 

Simon's algorithm~\cite{simon1997power} is central in quantum algorithm theory.
Historically, it was
an important milestone in the discovery by Shor of his celebrated quantum algorithm to solve integer factorization in polynomial time~\cite{DBLP:journals/siamcomp/Shor97}.
Interestingly, Simon's algorithm has also been applied in the context of symmetric cryptography. It was first used to break the 3-round Feistel construction~\cite{5513654} and 
then to prove that the Even-Mansour construction~\cite{6400943} is insecure with superposition queries.
While Simon's problem (which is the problem solved with Simon's
algorithm) might seem artificial at first sight, it appears  in certain
constructions in symmetric cryptography, in which ciphers and modes typically involve a lot of structure.

These first results, although quite striking, are not sufficient for evaluating the security of actual ciphers. Indeed, the confidence we have on symmetric ciphers depends on the amount of cryptanalysis that was performed on the primitive.
Only this effort allows researchers to define the security margin which measures how far  the construction is from being broken. Thanks to the large and always updated cryptanalysis toolbox built over the years in the \emph{classical} world, we have solid evaluations of the security of the primitives against classical adversaries. This is, however, no longer the case in the post-quantum world, \emph{i.e.} when considering quantum adversaries.

We therefore need to build a complete cryptanalysis toolbox for quantum adversaries, similar to what has been done for the classical world. This is a fundamental step in order to correctly evaluate the post-quantum security of current ciphers and to design new secure ciphers for the post-quantum world.

\subsubsection{Our results.}

We make progresses in this direction, and open new surprising and important ranges of applications for Simon's algorithm in symmetric cryptography:
\begin{enumerate}
\item The original formulation of Simon's algorithm is for functions whose collisions
happen only at some hidden period. We extend it to functions that have 
more collisions. This leads to 
a better analysis of previous applications of Simon's algorithm
  in symmetric cryptography.
\item We then show an attack against the LRW construction, used to turn
  a block-cipher into a tweakable block
  cipher~\cite{DBLP:journals/joc/LiskovRW11}.  Like the results
  on 3-round Feistel and Even-Mansour, this is an example of
  construction with provable security in the classical setting that becomes
  insecure against a quantum adversary.
\item Next, we study block cipher modes of operation.
We show that some of the most common modes for
  message authentication and authenticated encryption are completely broken in this setting.  We
  describe forgery attacks against standardized modes (CBC-MAC, PMAC,
  GMAC, GCM, and OCB), and against
  several CAESAR candidates, with complexity only $O(n)$, where $n$ is the size of the block.
In particular, this partially answers an open question by Boneh and
Zhandry~\cite{DBLP:conf/crypto/BonehZ13}: ``Do the CBC-MAC or
NMAC constructions give quantum-secure PRFs?''.

Those results are in stark contrast with a recent analysis of encryption
modes in the same setting: Anand \emph{et al.} show that
some classical encryption modes are secure against a quantum adversary
when using a quantum-secure
PRF~\cite{DBLP:conf/pqcrypto/AnandTTU16}. Our results imply that some authentication and authenticated encryption schemes
remain insecure with \emph{any} block cipher.
\item The last application is a quantization of slide attacks, a popular family of cryptanalysis that is  independent of the number of rounds of the attacked cipher. Our result is the first exponential speed-up obtained directly by a quantization of a classical cryptanalysis technique, with complexity dropping from $O(2^{n/2})$ to $O(n)$, where $n$ is the size of the block.
\end{enumerate}
These results imply that for the symmetric primitives we analyze, doubling the key length is not 
sufficient to restore security against quantum adversaries. 
A significant effort on quantum cryptanalysis of symmetric primitives is thus 
crucial for our long-term trust in these cryptosystems.

\subsubsection{The attack model.}

We consider attacks against classical cryptosystems using quantum resources.
This general setting broadly defines the field of post-quantum cryptography.
But attacking specific cryptosystems requires a more precise definition of the operations the adversary
is allowed to perform.
The simplest setting allows the adversary to perform local quantum computation.
For instance, this can be modeled by the quantum random oracle model, in which the adversary can query 
the oracle in an arbitrary superposition of the inputs~\cite{boneh11,brassard2011merkle,zhandry2015secure,unruh15}. A more practical setting allows
quantum queries to the hash function used to instantiate the oracle on a quantum computer.

We consider here a much stronger model in which,
in addition to local quantum operations, an adversary 
is granted an access to a possibly remote cryptographic oracle in superposition of the inputs, and
obtains the corresponding superposition of outputs.
In more detail, if the encryption oracle is described by a classical function $\mathcal{O}_k : \{0,1\}^n \to \{0,1\}^n$, then the adversary can make standard quantum queries
$|x\rangle |y\rangle \mapsto |x\rangle |\mathcal{O}_k(x) \oplus y\rangle$,
where $x$ and $y$ are arbitrary $n$-bit strings and $\ket x$, $|y\rangle$ are the corresponding $n$-qubit states expressed in the computational basis.  A circuit representing the oracle is given in Figure~\ref{fig:oracle}.
Moreover, any superposition $\sum_{x,y} \lambda_{x,y} \ket x \ket y$ is a valid input to the quantum oracle, who then
returns $\sum_{x,y} \lambda_{x,y} \ket x \ket {y\oplus O_k(x)}$.
In previous works, these attacks
have been called \emph{superposition
  attacks}~\cite{DBLP:conf/icits/DamgardFNS13}, \emph{quantum chosen
  message attacks}~\cite{DBLP:conf/crypto/BonehZ13} or \emph{quantum security}~\cite{DBLP:conf/focs/Zhandry12}.

Simon's algorithm requires the preparation of the uniform superposition of all $n$-bit strings,
 $\frac{1}{\sqrt{2^n}} \sum_x |x\rangle |0\rangle$\footnote{When there is no ambiguity, we
write $|0\rangle$ for the state $|0 \ldots 0\rangle$ of appropriate length.}.
For this input, the quantum encryption oracle returns
$ \frac{1}{\sqrt{2^n}} \sum_x |x\rangle |\mathcal{O}_k(x)\rangle$,
the superposition of all possible pairs of plaintext-ciphertext.
It might seem at first that this model gives an overwhelming power to the adversary and is therefore uninteresting. Note, however, that the laws of quantum mechanics imply that the measurement of such a $2n$-qubit state can only reveal $2n$ bits of information, making this model nontrivial.

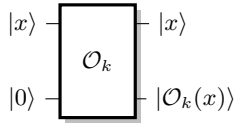
\begin{figure}[t]
	\centering
	\begin{tikzpicture}
	\node       (0inp) at (0,0) {$\ket 0$};
	\node	(xinp) at (0,1) {$\ket x$};
	\node	(0inpanc) at (0.7,0) {} edge[-] (0inp);
	\node	(xinpanc) at (0.7,1) {} edge[-] (xinp);
	\node[fun, minimum width=1cm, minimum height = 1.5cm]		(fun) at (1,0.5) {$\mathcal O_k$};
	\node	(xou2anc) at (1.4,0) {};
	\node	(xou1anc) at (1.4,1) {};
	\node	(xout1)  at (2,1) {$\ket x$} edge[-] (xou1anc);
	\node	(xout2)  at (2.3,0) {$\ket{\mathcal O_k(x)}$} edge[-] (xou2anc);
	\end{tikzpicture}
	\caption{The quantum cryptographic oracle.}
\end{figure}
\label{fig:oracle}

The simplicity of this model, together with the fact that it encompasses any reasonable model of quantum attacks
makes it very interesting.
For instance, \cite{DBLP:conf/eurocrypt/BonehZ13} 
gave constructions of message authenticated codes that remain secure against superposition attacks.
A similar approach was initiated by~\cite{DBLP:conf/icits/DamgardFNS13}, 
who showed how to construct secure multiparty protocols when an adversary
can corrupt the parties in superposition. A  protocol that is proven secure in this model may
truthfully be used in a quantum world.

Our work shows that superposition attacks, although they are not trivial, allow
new powerful strategies for the adversary.
Modes of operation that are provably secure against classical attacks
can then be broken. There exist a few options
to prevent the attacks that we present here.
A possibility is to forbid all kind of quantum access to a cryptographic oracle.
In a world where quantum resources become available, this restriction requires a careful attention.
This can be achieved for example by performing a quantum measurement of any incoming
quantum query to the oracle. But this task involves meticulous
engineering of quantum devices whose outcome remains uncertain. Even information theoretically secure
quantum cryptography remains vulnerable to attacks on their implementations,
as shown by attacks on quantum key distribution~\cite{zhao2008quantum,lydersen2010hacking,xu2010experimental}.

A more realistic approach is to develop a set of protocols that remains secure
against superposition attacks. Another advantage
of this approach is that it also covers more advanced scenarios, for example when
an encryption device is given to the adversary as an obfuscated algorithm.
Our work shows how important it is to develop protocols that remain secure against
superposition attacks.

Regarding symmetric cryptanalysis, we have already mentioned the protocol
of Boneh and Zhandry for MACs that remains secure against superposition attacks.
In particular, we answer negatively to their question asking wether CBC-MAC is secure in their model. 
Generic quantum attacks against symmetric cryptosystems
have also been considered. 
For instance, \cite{DBLP:journals/corr/Kaplan14} studies the security of iterated block ciphers, and 
 Anand et al. investigated the security of various modes of operations for encryption 
against superposition attacks~\cite{DBLP:conf/pqcrypto/AnandTTU16}. They show that OFB and CTR remain secure, while CBC and CFB
are not secure in general (with attacks involving Simon's algorithm), but 
are secure if the underlying PRF is quantum secure. Recently, \cite{DBLP:journals/corr/KaplanLLN15} considers 
symmetric families of cryptanalysis, describing quantum versions of differential and linear attacks.

Cryptographic notions like indistinguishability or semantic security are well understood
in a classical world. However, they become difficult to formalize when considering
quantum adversaries. The quantum chosen message model
is a good framework to study these~\cite{gagliardoni2015semantic,broadbent2015quantum,alagic2016computational}.

In this paper, we consider forgery attacks: the goal of
the attacker is to forge a tag for some arbitrary message, without the
knowledge of the secret key. In a quantum setting, we follow the
EUF-qCMA security definition that was given by Boneh and Zhandry~\cite{DBLP:conf/eurocrypt/BonehZ13}. A message
authentication code is broken by a quantum existential forgery attack if after $q$ queries to the cryptographic oracle, the adversary
can generate at least $q+1$ valid messages with corresponding tags.

\subsubsection{Organization.}

The paper is organized as follows. First, Section~\ref{sec:algo} introduces Simon's algorithm and explains how to modify it in order to handle functions that only approximately satisfy Simon's promise. This variant seems more appropriate for symmetric cryptography and may be of independent interest. 
Section~\ref{sec:previous} summarizes known quantum attacks against various constructions in symmetric cryptography.   
Section~\ref{sec:LRW} presents the attack against the LRW constructions.
In Section~\ref{sec:modes}, we show how Simon's algorithm can be used to obtain devastating attacks on several widely used modes of operations: CBC-MAC, PMAC, GMAC, GCM, OCB, as well as several CAESAR candidates. Section~\ref{sec:slide} shows the application of the algorithm to slide attacks, providing an exponential speed-up. The paper ends in Section~\ref{sec:conc} with a conclusion, pointing out possible new directions and applications.

\section{Simon's algorithm and attack strategy}
\label{sec:algo}
In this section, we  present Simon's problem~\cite{simon1997power} and the quantum algorithm for efficiently solving it. 
The simplest version of our attacks directly exploits this algorithm in order to recover some secret value of the encryption algorithm. Previous works have already considered such attacks against 3-round Feistel schemes and the Even-Mansour construction (see Section \ref{sec:previous} for details). 

Unfortunately, it is not always possible to recast an attack in terms of Simon's problem. More precisely, Simon's problem is a promise problem, and in many cases, the relevant promise (that only a structured class of collisions can occur) is not satisfied, far from it in fact. We show in Theorem~\ref{th:approx} below that, however, these additional collisions do not lead to a significant increase of the complexity of our attacks.

\subsection{Simon's problem and algorithm}
\label{sec:algorithm}
We first describe Simon's problem, and then the quantum  algorithm for solving it. 
We refer the reader to the recent review by Montanaro and de Wolf on quantum property testing for various applications of this algorithm \cite{MW13}.
We assume here a basic knowledge of the quantum circuit model. We denote the addition and
multiplication in a field with $2^n$ elements by ``$\oplus$'' and
``$\cdot$'', respectively.

We consider that the access to the input of Simon's problem, a function $f$, is made by querying it. A classical query oracle is a function $x \mapsto f(x)$. To run Simon's algorithm, it is required that the function $f$ can be queried quantum-mechanically. More precisely, it is supposed that the algorithm can make arbitrary quantum superpositions of queries of the form $\ket x \ket 0 \mapsto \ket x \ket {f(x)}$. 

Simon's problem is the following:
\begin{framed}
\noindent
\textbf{Simon's problem}: Given a Boolean function $f: \{0,1\}^n \rightarrow \{0,1\}^n$ and the promise that there exists $s \in \{0,1\}^n$ such that for any $(x,y) \in \{0,1\}^n,$ $[f(x)=f(y)]\Leftrightarrow [x\oplus y \in \{0^n, s\}]$, the goal is to find s.
\end{framed}

This problem can be solved classically by searching for collisions. The optimal
time to solve it is therefore $\Theta (2^{n/2}).$
On the other hand, Simon's algorithm solves this problem with quantum complexity $O(n)$.
Recall that the Hadamard transform $H^{\otimes n}$ applied on an $n$-qubit state $|x\rangle$ for some $x \in \{0,1\}^n$ gives
$H^{\otimes n} |x\rangle = \frac{1}{\sqrt{2^n}} \sum_{y \in \{0,1\}^n} (-1)^{x\cdot y} |y\rangle,$
where $x \cdot y := x_1 y_1 \oplus \cdots \oplus x_n y_n$.
 
The algorithm repeats the following five quantum steps.
\begin{enumerate}
\item Starting with a $2n$-qubit state $|0\rangle |0\rangle$, one applies a Hadamard transform $H^{\otimes n}$ to the first register to obtain the quantum superposition 
$$\frac{1}{\sqrt{2^n}} \sum_{x \in \{0,1\}^n} |x\rangle |0\rangle.$$
\item A quantum query to the function $f$ maps this to the state
$$\frac{1}{\sqrt{2^n}} \sum_{x \in \{0,1\}^n} |x\rangle |f(x)\rangle.$$
\item Measuring the second register in the computational basis yields a value $f(z)$ and collapses the first register to the state:
$$
\frac{1}{\sqrt{2}} (|z\rangle + |z \oplus s\rangle).
$$
\item Applying again the Hadamard transform $H^{\otimes n}$ to the first register gives:
$$
\frac{1}{\sqrt{2}} \frac{1}{\sqrt{2^n}} \sum_{y \in \{0,1\}^n} (-1)^{y\cdot z}\left(1+ (-1)^{y\cdot s}  \right)|y\rangle.
$$
\item The vectors $y$ such that $y\cdot s=1$ have amplitude 0. Therefore,
measuring the state in the computational basis yields a random vector $y$
 such that $y \cdot s = 0$. 
\end{enumerate}
By repeating this subroutine $O(n)$ times, one obtains $n-1$ independent vectors orthogonal to $s$ with high probability, and
$s$ can be recovered using basic linear algebra. Theorem \ref{th:approx} gives the trade-off between the number of repetitions of the subroutine and the success probability of the algorithm.

\subsection{Dealing with unwanted collisions}\label{variant}

In our cryptanalysis scenario, it is not always the case that the promise of Simon's problem is perfectly satisfied. More precisely, by construction, there will always exist an $s$ such that $f(x) = f(x \oplus s)$ for any input $x$, but there might be many more collisions than those of this form. If the number of such unwanted collisions is too large, one might not be able to obtain a full rank linear system of equations from Simon's subroutine after $O(n)$ queries. 
Theorem \ref{th:approx} rules this out provided that $f$
does not have too many collisions of the form $f(x) = f(x \oplus t)$
for some $t \not\in \{0, s\}$.

For  $f: \{0,1\}^n \to \{0,1\}^n$ such that $f(x \oplus s) = f(x)$ for all $x$, consider 
\begin{align}
\varepsilon(f,s) = \max_{t \in \{0,1\}^n \setminus\{0,s\}} \mathrm{Pr}_x[f(x) = f(x\oplus t)].
\end{align}
This parameter quantifies how far the function is from satisfying
Simon's promise. For a random function, one expects $\varepsilon(f,s) =
\Theta(n 2^{-n})$, following the analysis of~\cite{DBLP:journals/jmc/DaemenR07}. On the other hand, for a constant function, $\varepsilon(f,s)= 1$ and it is impossible to recover $s$.

The following theorem, whose proof can be found in Appendix~\ref{app:proof}, shows the effect of unwanted collisions on the success probability of Simon's algorithm.
\begin{theorem}[Simon's algorithm with approximate promise]\label{th:approx}
If $\varepsilon(f,s) \leq p_0   <1$, then Simon's algorithm returns $s$ with $c n$ queries, with probability at least $1-\big(2 \big(\frac{1+p_0}{2} \big)^{\! c}\big)^n$.
\end{theorem}

In particular, choosing $c \geq 3/(1-p_0)$ ensures that the
error decreases exponentially with $n$.
To apply our results, it is therefore sufficient to prove that $\varepsilon(f,s)$ is bounded away from 1.

\ifnopromise
Finally, if we apply Simon's algorithm without any 
bound on $\varepsilon(f,s)$, we can not always recover $s$
unambiguously.  Still if we select a random value $t$ orthogonal to all
vectors $u_i$ returned by each step of the algorithm, $t$ satisfy
$f(x \oplus t) = f(x)$ with high probability.

\begin{theorem}[Simon's algorithm without promise]\label{th:nopromise}
After $cn$ steps of Simon's algorithm, if $t$ is
orthogonal to all vectors $u_i$ returned by each step of the algorithm,
then
\( \Pr_x[f(x \oplus t) = f(t)] \ge p_0 \) with probability at least 
$1-\big(2 \big(\frac{1+p_0}{2} \big)^{\! c}\big)^n$.
\end{theorem}

In particular, choosing $c \geq 3/(1-p_0)$ ensures that the probability
is exponentially close to 1.
\fi

\subsection{Attack strategy}

The general strategy behind our attacks exploiting Simon's algorithm is to start with the encryption oracle $E_k: \{0,1\}^n \to \{0,1\}^n$ and exhibit a new function~$f$ that satisfies Simon's promise with two additional properties: the adversary should be able to query $f$ in superposition if he has quantum oracle access to $E_k$, and the knowledge of the string $s$ should be sufficient to break the cryptographic scheme. In the following, this function is called Simon's function.

In most cases, our attacks correspond to a classical collision attack.
In particular, the value $s$ will usually be the difference in the
internal state after processing a fixed pair of messages $(\alpha_0,
\alpha_1)$, \emph{i.e.}  $s = E(\alpha_0) \oplus E(\alpha_1)$.  The
input of $f$ will be inserted into the state with the difference $s$ so
that $f(x) = f(x \oplus s)$.

In our work, this function $f$ is of the form:
\begin{align*}
  f^1: x \quad &\mapsto P(\widetilde{E}(x) + \widetilde{E}(x\oplus s)) \quad  \text{or,}\\
f^2:  b,x \quad &\mapsto
  \begin{cases}
    \widetilde{E}(x) & \text{if $b = 0$,}\\
    \widetilde{E}(x \oplus s) & \text{if $b = 1$},
  \end{cases}
\end{align*}
where $\widetilde{E}$ is a simple function obtained from $E_k$ and $P$ a permutation. 
It is immediate to see that $f^1$ and $f^2$ have periods $s$ for $f^1$ or $1 ||s$ for $f^2$. 

In most applications, Simon's function satisfies $f(x)=f(y)$ for
$y \oplus x \in \{0,s\}$, but also for additional inputs $x,y$. 
Theorem~\ref{th:approx} extends Simon's algorithm precisely to this case.
In particular, if the additional collisions of $f$ are random, then Simon's algorithm is successful.
When considering explicit constructions, we can not  in general prove that the unwanted collisions
\emph{are} random, but rather that they \emph{look random enough}.
In practice,  if the function $\varepsilon(f,s)$ is not bounded, then
some of the primitives used in the construction have are far from
ideal.  We can show that this happens with low probability, and would
imply an classical attack against the system.
Applying Theorem~\ref{th:approx} is not trivial, but it stretches the range of application of Simon’s algorithm far beyond its original version.

\subsubsection{Construction of Simon's functions.}

To make our attacks as clear as possible, we provide 
the diagrams of circuits computing the function $f$.  These
circuits use a little number of basic building blocks represented in
Figure~\ref{fig:blocks}.

In our attacks, we often use a pair of arbitrary constants $\alpha_0$ and $\alpha_1$.
The choice of the constant is indexed by a bit $b$. We denote by $U_\alpha$ the gate
that maps $b$ to $\alpha_b$ (See Figure~\ref{fig:1t1map}). 
For simplicity, we ignore here the
additional qubits required in practice to make the transform reversible through padding.

Although it is well known that arbitrary quantum states cannot be cloned, we use the \emph{CNOT} gate
to copy classical information. More precisely, a CNOT gate can
copy states in the computational basis:
\mbox{$CNOT: \ket x \ket 0 \to \ket x \ket x$}. This transform is represented in Figure~\ref{fig:cnot}.

Finally, any unitary transform $U$ can be controlled by a bit $b$. This operation, denoted $U^b$
maps $x$ to $U(x)$ if $b=1$ and leaves $x$ unchanged otherwise.
In the quantum setting, the qubit $\ket b$ can be in a superposition of $0$ and $1$,
resulting in a superposition of $\ket x$ and $\ket{U(x)}$.
The attacks that we present in the following sections only make use of
this procedure when the attacker knows a classical description
of the unitary to be controlled. In particular, 
we do not apply it to the cryptographic oracle.

When computing Simon's function, \emph{i.e.} the function $f$ on which Simon's algorithm is applied, 
the registers containing the value of $f$ must be unentangled with any other working register.
Otherwise, these registers, which might hinder the periodicity of the function, have to be
taken into account in Simon's algorithm and the whole procedure could fail.

\begin{figure}
  {
    \setcounter{subfig}{0}
    \makeatletter
    \long\def\@makecaption#1#2{%
      \small
      \stepcounter{subfig}%
      \vskip\abovecaptionskip
      \sbox\@tempboxa{{\bfseries #1.} #2}%
      \ifdim \wd\@tempboxa >\hsize
      {\bfseries #1.} #2\par
      \else
      \global \@minipagefalse
      \hb@xt@\hsize{\hfil\box\@tempboxa\hfil}%
      \fi
      \vskip\belowcaptionskip\addtocounter{figure}{-1}}
    \let\theoldfigure=\thefigure
    \def\thefigure{\theoldfigure.\intcalcAdd{\thesubfig}{1}}
    \def\fnum@figure{\theoldfigure.\thesubfig}
    \makeatother
    \centering
    {
    \begin{minipage}[t]{0.30\textwidth}
        \centering
	\begin{tikzpicture}
	\node       (0inp) at (0,0) {$\ket b$};
	\node[fun, minimum size=1cm]		(fun) at (1,0) {$U_\alpha$} edge[-] (0inp);
	\node	(xout2)  at (2.2,0) {$\ket {\alpha_b}$} edge[-] (fun);
	\end{tikzpicture}
        \caption{One-to-one mapping.}
        \label{fig:1t1map}
      \end{minipage}
    }%
    \hfill%
    {
    \begin{minipage}[t]{0.30\textwidth}
        \centering
	\begin{tikzpicture}
	\node       (0inp) at (0,0) {$\ket 0$};
	\node	(xinp) at (0,1) {$\ket x$};
	\node[fork]	(control) at (1,1) {} edge[-] (xinp);
	\node[xor]		(xor) at (1,0) {} edge[-] (0inp) edge[-] (control);
	\node	(xout1)  at (2,1) {$\ket x$} edge[-] (control);
	\node	(xout2)  at (2,0) {$\ket x$} edge[-] (xor);
	
	\end{tikzpicture} 
        \caption{CNOT gate.}
        \label{fig:cnot}
      \end{minipage}
    }%
    \hfill%
    {
            \begin{minipage}[t]{0.30\textwidth}
        \centering
	\begin{tikzpicture}
	\node       (0inp) at (0,0) {$\ket x$};
	\node	(xinp) at (0,1) {$\ket b$};
	\node[fork]	(control) at (1,1) {} edge[-] (xinp);
	\node[fun, minimum size=1cm]		(fun) at (1,0) {$U$} edge[-] (0inp) edge[-] (control);
	\node	(xout1)  at (2,1) {$\ket b$} edge[-] (control);
	\node	(xout2)  at (2.3,0) {$\ket {U^b(x)}$} edge[-] (fun);
	
	\end{tikzpicture}        
        \caption{Controlled Unitary.}
        \label{fig:controlledU}
      \end{minipage}
    }%
  }
  \caption{Circuit representation of basic building blocks.}\label{fig:blocks}
\end{figure}

\section{Previous works}
\label{sec:previous}

Previous works have used Simon's algorithm to break
the security of classical constructions in symmetric cryptography:
the Even-Mansour construction and the 3-round Feistel scheme.  We now
explain how these attacks work with our terminology and extend two of
the results. First, we show that the attack on the Feistel scheme can
be extended to work with random functions, where the original analysis held only
for random permutations. Second, using our analysis Simon's algorithm with approximate
promise, we make the number of queries required to attack the Even-Mansour construction
more precise.
These observations have been independently made by Santoli and Schaffner~\cite{SS}.
They use a slightly different approach, which consists in analyzing the run of Simon's algorithm for these
specific cases.

\subsection{Applications to a three-round Feistel scheme}
\label{sec:feistel}

The Feistel scheme is a classical construction to build a random
permutation out of random functions or random permutations.  In a seminal work, Luby and
Rackoff proved that a three-round Feistel scheme
is a secure pseudo-random
permutation~\cite{DBLP:journals/siamcomp/LubyR88}.

A three-round Feistel scheme with input $(x_L,x_R)$ and output $(y_L,
y_R)= E(x_L, x_R)$ is built from three round functions $R_1$, $R_2$, $R_3$ as (see Figure~\ref{fig:feistel}):
\begin{align*}
  (u_0, v_0) &= (x_L,x_R), &
  (u_{i}, v_{i}) &= (v_{i-1} \oplus R_i(u_{i-1}), u_{i-1}), &
  (y_L, y_R)   &= (u_3, v_3).
\end{align*}

\begin{figure}
  \begin{minipage}[b]{0.30\textwidth}
    \centering
    \begin{tikzpicture}[yscale=1.2]
      \node (xl) at (0,3.5) {$x_L$};
      \node (xr) at (2,3.5) {$x_R$};
      
      \node[fork] (y1) at (0,3) {};
      \node[fun]  (f1) at (1,3) {$R_1$};
      \node[xor]  (x1) at (2,3) {};

      \node[fork] (y2) at (0,2) {};
      \node[fun]  (f2) at (1,2) {$R_2$};
      \node[xor]  (x2) at (2,2) {};

      \node[fork] (y3) at (0,1) {};
      \node[fun]  (f3) at (1,1) {$R_3$};
      \node[xor]  (x3) at (2,1) {};

      \node (yl) at (0,0) {$y_L$};
      \node (yr) at (2,0) {$y_R$};

      \draw[->] (xr) -- (x1);
      \draw[->] (xl) -- (y1) -- (0,2.66) -- (2,2.33) -- (x2);

      \draw[->] (y1) -- (f1) -- (x1);
      \draw[->] (x1) -- (2,2.66) -- (0,2.33) -- (y2) -- (0,1.66) -- (2,1.33) -- (x3);

      \draw[->] (y2) -- (f2) -- (x2);
      \draw[->] (x2) -- (2,1.66) -- (0,1.33) -- (y3) -- (0,.66) -- (2,.33) -- (yr);

      \draw[->] (y3) -- (f3) -- (x3);
      \draw[->] (x3) -- (2,.66) -- (0,.33) -- (yl);

    \end{tikzpicture}
    
    \caption{Three-round Feistel scheme.}
    \label{fig:feistel}
  \end{minipage}
  \hfill%
  \begin{minipage}[b]{0.65\textwidth}
    \centering
    \begin{tikzpicture}
      \node[left]       (0in2) at (0,2) {$\ket b$};
      \node[left]       (0in1) at (0,1) {$\ket x$};
      \node[left]       (0in0) at (0,0) {$\ket 0$};
      \node[fun, minimum size=1cm]  (U1) at (1,2) {$U_{\alpha}$} edge[-] (0in2);
      \node[fun,minimum height = 2.5cm, minimum width = 1cm]  (F1) at (3,1) {$y_R$};
      \draw	(U1) -- (U1 -| F1.west);
      \draw	(0in1) -- (0in1 -| F1.west);
      \draw	(0in0) -- (0in0 -| F1.west);
      \node[fork]	(control) at (4,2) {} edge[-] (control -| F1.east) ;
      \node[xor]		(xor) at (4,0) {} edge[-] (control) edge[-] (xor -| F1.east);
      \node[fun, minimum size=1cm]  (U2) at (5,2) {$U_{\alpha}^{-1}$} edge[-] (U2 -| F1.east);
      \node[right]       (0out2) at (6,2) {$\ket b$} edge[-] (U2);
      \node[right]       (0out1) at (6,1) {$\ket x$} edge[-] (0out1 -| F1.east);
      \node[right]       (0out0) at (6,0) {$\ket{f(b,x)}$} edge[-] (xor);
    \end{tikzpicture}
    \caption[Simon's function for Feistel]{Simon's function for
      Feistel.\\ \hspace*{\linewidth}} 
    \label{feistelfunc}
  \end{minipage}
\end{figure}

In order to distinguish a Feistel scheme from a random permutation in a
quantum setting, Kuwakado and Morii~\cite{5513654} consider the case
were the $R_i$ are permutations, and define the following
function, with two arbitrary constants $\alpha_0$ and $\alpha_1$  such that $\alpha_0 \ne \alpha_1$:
\begin{align*}
  f: \{0,1\} \times \{0,1\}^{n} &\to \{0,1\}^{n} \\
  b,x \quad &\mapsto y_R \oplus \alpha_b, \quad \text{where $(y_R, y_L) = E(\alpha_b,x)$}\\
  f(b,x) &= R_2(x \oplus R_1(\alpha_b))
\end{align*}
In particular, this $f$ satisfies
$f(b,x) = f(b \oplus 1, x \oplus R_1(\alpha_0) \oplus R_1(\alpha_1))$.
Moreover,
\begin{align*}
f(b',x') = f(b,x)
&\Leftrightarrow x' \oplus R_1(\alpha_{b'}) = x \oplus R_1(\alpha_b)\\
&\Leftrightarrow
  \begin{cases}
      x' \oplus x = 0 & \text{if $b' = b$}\\
      x' \oplus x = R_1(\alpha_0)\oplus R_1(\alpha_1) & \text{if $b' \ne b$}
  \end{cases}
\end{align*}
Therefore, the function satisfies Simon's promise with $s =
1\|R_1(\alpha_0) \oplus R_1(\alpha_1)$, and we can recover $R_1(\alpha_0)
\oplus R_1(\alpha_1)$ using Simon's algorithm.  This gives a
distinguisher, because Simon's algorithm applied to a random permutation
returns zero with high probability.
\ifnopromise
This can be seen from Theorem 2, using the fact that with overwhelming
probability\cite{DBLP:journals/jmc/DaemenR07}, there is no value $t \ne 0$ such that $\Pr_x[f(x \oplus t)
= f(x)] > 1/2$ for a random permutation $f$.
\fi

We can also verify that the value $R_1(\alpha_0) \oplus R_1(\alpha_1)$
is correct with two additional classical queries $(y_L, y_R) = E(\alpha_0, x)$ and
$(y'_L, y'_R) = E(\alpha_1, x \oplus R_1(\alpha_0)
\oplus R_1(\alpha_1))$ for a random $x$.  If the value is
correct, we have $y_R \oplus y'_R = \alpha_0
\oplus \alpha_1$.

Note that in their attack, Kuwakado and Morii implicitly assume that the adversary can query in superposition
an oracle that returns solely the left part $y_L$ of the encryption.
If the adversary only has access to the complete encryption oracle $E$, then a query in superposition
would return two \emph{entangled} registers containing the left and right parts, respectively.
In principle, Simon's algorithm requires the register containing the input value to
be completely disentangled from the others.

\subsubsection{Feistel scheme with random functions.}  
Kuwakado and Morii~\cite{5513654} analyze only the case where the round
functions $R_i$ are permutations.  We now extend this analysis to \emph{random functions} $R_i$.  The function $f$ defined above still satisfies $f(b,x) = f(b
\oplus 1, x \oplus R_1(\alpha_0) \oplus R_1(\alpha_1))$, but it doesn't
satisfy the exact promise of Simon's algorithm: there are additional collisions in
$f$, between inputs with random differences.
However, the previous distinguisher is still valid: at the end of Simon's
algorithm, there exist at least one non-zero value orthogonal to all the
values $y$ measured at each step: $s$.  This would not be the case with
a random permutation.

Moreover, we can show that $\varepsilon(f,1\|s)<1/2$ with overwhelming
probability, so that Simon's algorithm still recovers $1\|s$ following
Theorem~\ref{th:approx}.
If $\varepsilon(f,1\|s) > 1/2$, there exists $(\tau,t)$ with $(\tau,t)
\not\in \{(0,0), (1,s)\}$ such that:
  $\mathrm{Pr}[f(b,x) = f(b\oplus \tau, x\oplus t)] > 1/2$.
Assume first that $\tau = 0$, this implies:
\begin{align*}
    \mathrm{Pr}[f(0,x) = f(0, x\oplus t)] > 1/2  \quad \text{or} \quad 
    \mathrm{Pr}[f(1,x) = f(1, x\oplus t)] > 1/2.
\end{align*}
Therefore, for some $b$, 
  $\Pr[R_2(x \oplus R_1(\alpha_b)) = R_2(x \oplus t \oplus R_1(\alpha_b))] > 1/2$, 
 \emph{i.e.} $\Pr[R_2(x) = R_2(x \oplus t)] > 1/2$.
Similarly, if $\tau = 1$,
  $\Pr[R_2(x \oplus R_1(\alpha_0)) = R_2(x \oplus t \oplus R_1(\alpha_1))] > 1/2$,
\emph{i.e.} $ \Pr[R_2(x) = R_2(x \oplus t \oplus R_1(\alpha_0) \oplus R_1(\alpha_1))] > 1/2$.

To summarize, if $\varepsilon(f,1\|s) > 1/2$, there exists $u \ne 0$
such that $\Pr[R_2(x) = R_2(x \oplus u)] > 1/2$.  This only happens with
negligible probability for a random choice of
$R_2$ as shown in~\cite{DBLP:journals/jmc/DaemenR07}.

\subsection{Application to the Even-Mansour construction}
\label{sec:even-mansour}

The Even-Mansour construction is a simple construction to build a
block cipher from a public permutation~\cite{DBLP:journals/joc/EvenM97}.
For some permutation $P$, the cipher is:
\[ E_{k_1,k_2}(x) = P(x \oplus k_1) \oplus k_2. \]
Even and Mansour have shown that this construction is secure in the
random permutation model, up to $2^{n/2}$ queries, where $n$ is the
size of the input to $P$.

\begin{figure}
  \begin{minipage}[b]{0.35\textwidth}
  \centering
  \begin{tikzpicture}[x={(0pt,-.75cm)},y={(-1cm,0pt)}]
      \node       (P0) at (-0.5,0) {$x$};
      \node      (K1) at (0.5,1) {$k_1$};
      \node[xor] (x0) at (0.5,0) {} edge[<-] (P0) edge[<-] (K1);
      \node[fun,minimum size=1cm] (p0) at (1.8,0) {$P$} edge[<-] (x0);
      \node      (k2) at (3,1) {$k_2$};
      \node[xor] (x1) at (3,0) {} edge[<-] (p0) edge[<-] (k2);
      \node	(out) at (4,0) {$E_{k_1,k_2}(x)$} edge[<-] (x1);
    \end{tikzpicture}
    \caption{Even-Mansour scheme.}
  \end{minipage}
  \hfill%
  \begin{minipage}[b]{0.60\textwidth}
    \centering
    \begin{tikzpicture}[xscale=0.66]
      \node[left] (a) at (1,0) {$|0\rangle$};
      \node[left] (b) at (1,1) {$|x\rangle$};

      \node (c) at (4,0) {$|E_k(x)\rangle$} edge (a);
      \node (d) at (4,1) {$|x\rangle$} edge (b);

      \node[right] (e) at (7,0) {$|E_k(x) \oplus P(x)\rangle$} edge (c);
      \node[right] (f) at (7,1) {$|x\rangle$} edge (d);

      \node[fun,minimum height=2cm, minimum width=1cm] (f) at (2,.5) {$E_k$};
      \node[fun,minimum height=2cm, minimum width=1cm] (f) at (6,.5) {$P$};
    \end{tikzpicture}
    \caption{Simon's function for Even-Mansour.}
  \end{minipage}
\end{figure}

However, Kuwakado and Morii~\cite{6400943} have shown that the security
of this construction collapses if an adversary can query an encryption
oracle with a superposition of states.  More precisely, they define the
following function:
\begin{align*}
  f: \{0,1\}^{n} &\to \{0,1\}^{n} \\
  x \quad &\mapsto E_{k_1,k_2}(x) \oplus P(x)= P(x \oplus k_1) \oplus P(x) \oplus k_2.
  \end{align*}
In particular, $f$ satisfies $f(x \oplus k_1) = f(x)$ (interestingly, 
the slide with a twist attack of Biryukov and Wagner\cite{DBLP:conf/eurocrypt/BiryukovW00} uses the same property).  
However, there are additional collisions in $f$
between inputs with random differences. As in the attack against the
Feistel scheme with random round functions, we
use Theorem~\ref{th:approx}, to show that Simon's algorithm recovers 
$k_1$\footnote{Note that Kuwakado and Morii just assume that each
  step of Simon's algorithm gives a random vector orthogonal to $k_1$.
  Our analysis is more formal and captures the conditions on $P$
  required for the algorithm to be successful.}.

We show that $\varepsilon(f,k_1)<1/2$ with overwhelming probability for
a random permutation $P$, and if $\varepsilon(f,k_1)>1/2$, then there
exists a classical attack against the Even-Mansour scheme.
Assume that $\varepsilon(f,k_1) > 1/2$, that is,
there exists $t$ with $t \not\in \{0, k_1\}$ such that $ \mathrm{Pr}[f(x) = f(x\oplus t)] > 1/2$, \emph{i.e.},
\begin{align*}
   p=\mathrm{Pr}[P(x) \oplus P(x\oplus k_1)  \oplus P(x\oplus t)  \oplus
  P(x\oplus t\oplus k_1) = 0] &> 1/2.
\end{align*}
This correspond to higher order differential for $P$ with probability
$1/2$, which only happens with negligible probability for a random choice
of $P$.  In addition, this would imply the existence of a simple
classical attack against the scheme:
\begin{enumerate}
\item Query $y = E_{k_1,k_2}(x)$ and $y' = E_{k_1,k_2}(x \oplus t)$
\item Then $y \oplus y' = P(x) \oplus P(x \oplus t)$ with probability at
  least one half
\end{enumerate}
Therefore, for any instantiation of the Even-Mansour scheme with a fixed
$P$,
either there exist a classical distinguishing attack (this only happens with
negligible probability with a random $P$), or Simon's algorithm successfully recovers
$k_1$.  In the second case, the value of $k_2$ can then be recovered
from an additional classical query: $k_2 = E(x) \oplus P(x \oplus k_1)$.

\medskip

In the next sections, we give new applications of Simon's algorithm, to
break various symmetric cryptography schemes.

\section{Application to the LRW construction}
\label{sec:LRW}

\begin{figure}
  {
    \setcounter{subfig}{0}
    \makeatletter
    \long\def\@makecaption#1#2{%
      \small
      \stepcounter{subfig}%
      \vskip\abovecaptionskip
      \sbox\@tempboxa{{\bfseries #1.} #2}%
      \ifdim \wd\@tempboxa >\hsize
      {\bfseries #1.} #2\par
      \else
      \global \@minipagefalse
      \hb@xt@\hsize{\hfil\box\@tempboxa\hfil}%
      \fi
      \vskip\belowcaptionskip\addtocounter{figure}{-1}}
    \let\theoldfigure=\thefigure
    \def\thefigure{\theoldfigure.\intcalcAdd{\thesubfig}{1}}
    \def\fnum@figure{\theoldfigure.\thesubfig}
    \makeatother
    \centering
    {
      \begin{minipage}[t]{0.30\textwidth}
        \centering
        \begin{tikzpicture}
          \node[fun,minimum size=1cm] (e) at (0,0) {$E_k$};
          \node[xor] (x) at (0,1) {};
          \node[xor] (y) at (0,-1) {};
          \node      (p) at (0,2) {$p$};
          \node      (c) at (0,-2) {$c$};
          \node      (t) at (-2,0) {$t$};
          \node[draw,minimum height=1cm] (h) at (-1,0) {$h$};

          \draw[->] (p) -- (x);
          \draw[->] (x) -- (e);
          \draw[->] (e) -- (y);
          \draw[->] (y) -- (c);
          \draw[->] (h) |- (x);
          \draw[->] (h) |- (y);
          \draw[->] (t) -- (h);
        \end{tikzpicture}
        
        \caption{LRW construction.}
        \label{fig:lrw}
      \end{minipage}
    }%
    \hfill%
    {
      \begin{minipage}[t]{0.30\textwidth}
        \centering
        \begin{tikzpicture}
          \node[fun,minimum size=1cm] (e) at (0,0) {$E_k$};
          \node[xor] (x) at (0,1) {};
          \node[xor] (y) at (0,-1) {};
          \node      (p) at (0,2) {$p$};
          \node      (c) at (0,-2) {$c$};
          \node      (t1) at (-2,1) {$2^t \cdot L$};
          \node[fork] (f) at (-1,1) {};

          \draw[->] (p) -- (x);
          \draw[->] (x) -- (e);
          \draw[->] (e) -- (y);
          \draw[->] (y) -- (c);
          \draw[->] (t1) -- (f) -- (x);
          \draw[->] (f) |- (y);
        \end{tikzpicture}
        
        \caption[XEX construction]{XEX construction.}
        \label{fig:xex}
      \end{minipage}
    }%
    \hfill%
    {
      \begin{minipage}[t]{0.30\textwidth}
        \centering
        \begin{tikzpicture}
          \node[fun,minimum size=1cm] (e) at (0,0) {$E_k$};
          \node[xor] (x) at (0,1) {};
          \node      (p) at (0,2) {$p$};
          \node      (c) at (0,-2) {$c$};
          \node      (t1) at (-2,1) {$2^t \cdot L$};

          \draw[->] (p) -- (x);
          \draw[->] (x) -- (e);
          \draw[->] (e) -- (c);
          \draw[->] (t1) -- (x);
        \end{tikzpicture}
        
        \caption[XE construction]{XE construction.}
        \label{fig:xe}
      \end{minipage}
    }%
  }
  \caption{The LRW construction, and efficient instantiations XEX (CCA secure)
    and XE (only CPA secure).}
  \label{fig:LRW}
\end{figure}

We now show a new application of Simon's algorithm to the LRW construction.
The LRW construction, introduced by Liskov, Rivest and
Wagner~\cite{DBLP:journals/joc/LiskovRW11}, turns a block cipher into a
tweakable block cipher, \emph{i.e.} a family of unrelated block ciphers.
The tweakable block cipher is a very useful primitive to build modes for
encryption, authentication, or authenticated encryption.  In particular,
tweakable block ciphers and the LRW construction were inspired by the
first version of OCB, and later versions of OCB use the tweakable block
ciphers formalism.  The LRW construction uses a (almost) universal hash
function $h$ (which is part of the key), and is defined as (see also Figure~\ref{fig:LRW}):
\[\widetilde{E}_{t,k}(x) = E_k(x \oplus h(t)) \oplus h(t).\]

We now show that the LRW construction is not secure in a quantum
setting.  We fix two arbitrary tweaks $t_0, t_1$, with $t_0 \ne t_1$,
and we define the following function:
\begin{align*}
  f: \{0,1\}^{n} &\to \{0,1\}^{n} \\
  x \quad &\mapsto \widetilde{E}_{t_0, k}(x) \oplus \widetilde{E}_{t_1,k}(x) \\
f(x)&= E_k\big(x \oplus h(t_0)\big) \oplus h(t_0) \oplus E_k\big(x \oplus h(t_1)\big) \oplus h(t_1).
\end{align*}
Given a superposition access to an oracle for an LRW tweakable block cipher,
we can build a circuit implementing this function, using the
construction given in Figure~\ref{fig:lrwcirc}. In the circuit,
the cryptographic oracle $\widetilde E_{t,k}$ takes two inputs: the block $x$ to be encrypted
and the tweak $t$. Since the tweak comes out of $\widetilde E_{t,k}$
unentangled with the other register,
we do not represent this output in the diagram. In practice, the output
is forgotten by the attacker.

It is easy to see that
this function satisfies $  f(x)   = f(x \oplus s)$ with $s = h(t_0) \oplus h(t_1)$.
Furthermore, the quantity
$\varepsilon(f,s) = \max_{t \in \{0,1\}^n \setminus\{0,s\}}
\mathrm{Pr}[f(x) = f(x\oplus t)]$ is bounded with overwhelming
probability, assuming that $E_k$ behaves as a random permutation.
Indeed if $\varepsilon(f,s) > 1/2$,
there exists some $t$ with $t \not\in \{0, s\}$ such that $\mathrm{Pr}[f(x) = f(x\oplus t)] > 1/2$, \emph{i.e.},
\[
  \mathrm{Pr}[
E_k\big(x\big) \oplus E_k\big(x \oplus s\big) \oplus
E_k\big(x \oplus t)\big) \oplus E_k\big(x \oplus t
\oplus s\big) = 0]
> 1/2
\]
This correspond to higher order differential for $E_k$ with probability
$1/2$, which only happens with negligible probability for a random permutation.
Therefore, if $E$ is a pseudo-random permutation family,
$\varepsilon(f,s) \le 1/2$ with overwhelming probability, and running Simon's algorithm with the function $f$ returns
$h(t_0) \oplus h(t_1)$.  The assumption that $E$ behaves as a PRP family is
required for the security proof of LRW, so it is reasonable to make the
same assumption in an attack.  More concretely, a block cipher with a
higher order differential with probability $1/2$ as seen above would
probably be broken by classical attacks.  The attack is not immediate
because the differential can depend on the key, but it would seem to
indicate a structural weakness.
\ifnopromise
In the following sections, some attacks can also be mounted using
Theorem~\ref{th:nopromise} without any assumptions on $E$.
\fi

In any case, there exist at least one non-zero value orthogonal to all
the values $y$ measured during Simon's algorithm: $s$.  This would not
be the case if $f$ is a random function, which gives a distinguisher
between the LRW construction and an ideal tweakable block cipher with
$O(n)$ quantum queries to~$\widetilde{E}$.

In practice, most instantiations of LRW use a finite field
multiplication to define the universal hash function $h$, with a secret
offset $L$ (usually computed as $L = E_k(0)$).  Two popular
constructions are:
\begin{itemize}
\item $h(t) = \gamma(t) \cdot L$, used in OCB1~\cite{DBLP:conf/ccs/RogawayBBK01}, OCB3~\cite{DBLP:conf/fse/KrovetzR11} and PMAC~\cite{DBLP:conf/eurocrypt/BlackR02}, with a
  Gray encoding $\gamma$ of $t$,
\item $h(t) = 2^t \cdot L$, the XEX construction, used in OCB2~\cite{DBLP:conf/asiacrypt/Rogaway04}.
\end{itemize}
In both cases, we can recover $L$ from the value $h(t_0) \oplus h(t_1)$
given by the attack.

This attack is important, because many recent modes of operation
are inspired by the LRW construction, and the XE and XEX instantiations,
such as CAESAR candidates AEZ~\cite{DBLP:conf/eurocrypt/HoangKR15},
COPA~\cite{DBLP:conf/asiacrypt/AndreevaBLMTY13},
OCB~\cite{DBLP:conf/fse/KrovetzR11},
OTR~\cite{DBLP:conf/eurocrypt/Minematsu14},
Minalpher~\cite{CAESAR_Minalpher},
OMD~\cite{DBLP:conf/sacrypt/CoglianiMNCRVV14}, and
POET~\cite{DBLP:conf/fse/AbedFFLLMW14}.  We will see in the next
section that variants of this attack can be applied to each of these
modes.

\begin{figure}
	\centering
	\begin{tikzpicture}
	\node       (b) at (0,2) {$\ket 0$};
	\node       (x) at (0,1) {$\ket x$};
	\node[fun, minimum size=1cm]		(alphab) at (1,2) {$U_t$} edge[-] (b);
	\node	(binpCBC) at (2.7,2) {} edge[-] (alphab);
	\node	(xinpO1) at (2.7,1) {} edge[-] (x);
	\node       (0) at (0,0) {$\ket 0$};
	\node	(0inpO1) at (2.7,0) {} edge[-] (0);	
	\node[fun, minimum width=1cm, minimum height=3cm] at (3,1) {$\widetilde E_{t_0,k}$};
	\node	(boutO1) at (4,2) {$\ket 1$};
	\node	(xoutO1) at (3.4,1) {};
	\node	(x0outO1) at (3.4,0) {};
	\node[fun, minimum size=1cm]		(oplusu) at (5.25,2) {$U_t$} edge[-] (boutO1);
	\node	(0inpO2) at (6.7,0) {} edge[-] (x0outO1);
	\node	(1inpO2) at (6.7,1) {} edge[-] (xoutO1);
	\node	(2inpO2) at (6.7,2) {} edge[-] (oplusu);
	\node[fun, minimum width=1cm, minimum height=3cm] at (7,1) {$\widetilde E_{t_1,k}$};
	\node	(0outO2) at (7.4,0) {};
	\node	(1outO2) at (7.4,1) {};
	\node	(2outO2) at (7.4,2) {};
	\node	(0out) at (9,0) {$\ket {f(x)}$} edge[-] (0outO2);
	\node	(bout) at (8.5,1) {$\ket x$} edge[-] (1outO2);
	\end{tikzpicture}
\caption{Simon's function for LRW.}\label{fig:lrwcirc}
\end{figure}
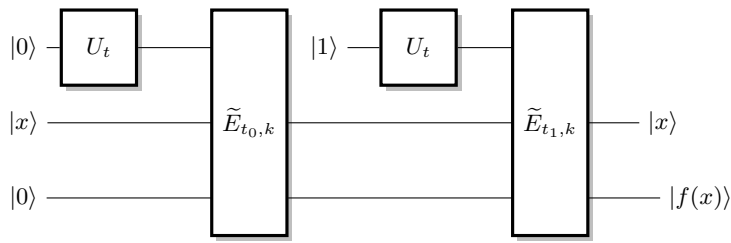

\section{Application to block cipher modes of operations}
\label{sec:modes}

We now give new applications of Simon's algorithm to the security of
block cipher modes of operations.  In particular, we show how to break
the most popular and widely used block-cipher based MACs, and message
authentication schemes: CBC-MAC (including variants such as XCBC~\cite{DBLP:conf/crypto/BlackR00},
OMAC~\cite{DBLP:conf/fse/IwataK03}, and CMAC~\cite{FIPS-800-38B}), GMAC~\cite{DBLP:conf/indocrypt/McGrewV04}, PMAC~\cite{DBLP:conf/eurocrypt/BlackR02}, GCM~\cite{DBLP:conf/indocrypt/McGrewV04} and OCB~\cite{DBLP:conf/fse/KrovetzR11}.  We also show attacks
against several CAESAR candidates.  In each case, the mode is proven
secure up to $2^{n/2}$ in the classical setting, but we show how, by a
reduction to Simon's problem, forgery attacks can be performed with
superposition queries at a cost of $O(n)$.

\textbf{Notations and preliminaries.}
We consider a block cipher $E_k$, acting on blocks of length $n$,
where the subscript $k$ denotes the key.  For simplicity, we only
describe the modes with full-block messages, the attacks can trivially
be extended to the more general modes with arbitrary inputs.  In
general, we consider a message $M$ divided into $\ell$ $n$-bits block:
$M=m_1 \| \ldots \| m_\ell$.  We also assume that the MAC is not
truncated, \emph{i.e.} the output size is $n$ bits.  In most cases,
the attacks can be adapted to truncated MACS.

\subsection{Deterministic MACs: CBC-MAC and PMAC}

We start with deterministic Message Authentication Codes, or MACs.  A
MAC is used to guarantee the authenticity of messages, and should
be immune against forgery attacks.  The standard security model is that it
should be hard to forge a message with a valid tag, even given access
to an oracle that computes the MAC of any chosen message (of course
the forged message must not have been queried to the oracle).

To translate this security notion to the quantum setting, we assume
that the adversary is given an oracle that takes a quantum
superposition of messages as input, and computes the superposition of
the corresponding MAC.

\subsubsection{CBC-MAC.}

CBC-MAC is one of the first MAC constructions, inspired by the CBC
encryption mode.  Since the basic CBC-MAC is only secure when the
queries are prefix-free, there are many variants of CBC-MAC to provide
security for arbitrary messages.  In the following we describe the
Encrypted-CBC-MAC variant~\cite{DBLP:journals/jcss/BellareKR00}, using two keys $k$ and $k'$, but the attack
can be easily adapted to other variants \cite{DBLP:conf/crypto/BlackR00,DBLP:conf/fse/IwataK03,FIPS-800-38B}.  On a message
$M=m_1 \| \ldots \| m_\ell$, CBC-MAC is defined as (see Figure~\ref{fig:cbc-mac}):
\begin{align*}
  x_0 &= 0&
  x_i &= E_{k}(x_{i-1} \oplus m_i)&
  \CBCMAC(M) = E_{k'}(x_{\ell})
\end{align*}
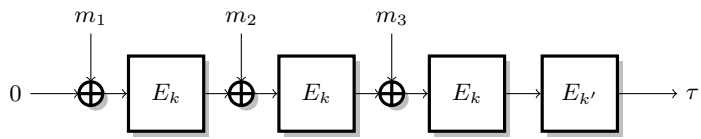
\begin{figure}
  \centering
  \begin{tikzpicture}
      \node       (k) at (0,0) {0};
      \node      (m0) at (1,1) {$m_1$};
      \node[xor] (x0) at (1,0) {} edge[<-] (k) edge[<-] (m0);
      \node[fun,minimum size=1cm] (p0) at (2,0) {$E_k$} edge[<-] (x0);
      \node      (m1) at (3,1) {$m_2$};
      \node[xor] (x1) at (3,0) {} edge[<-] (p0) edge[<-] (m1);
      \node[fun,minimum size=1cm] (p1) at (4,0) {$E_k$} edge[<-] (x1);
      \node      (m2) at (5,1) {$m_3$};
      \node[xor] (x2) at (5,0) {} edge[<-] (p1) edge[<-] (m2);
      \node[fun,minimum size=1cm] (p2) at (6,0) {$E_k$} edge[<-] (x2);
      \node[fun,minimum size=1cm] (p3) at (7.5,0) {$E_{k'}$} edge[<-] (p2);
      \node      (t)  at (9,0) {$\tau$} edge[<-] (p3);
  \end{tikzpicture}
  
  \caption{Encrypt-last-block CBC-MAC.}
  \label{fig:cbc-mac}
\end{figure}

CBC-MAC is standardized and widely used.  It has been proved to be
secure up to the birthday bound~\cite{DBLP:journals/jcss/BellareKR00},
assuming that the block cipher is indistinguishable from a random
permutation.
 
\subsubsection{Attack.}

We can build a powerful forgery attack on CBC-MAC with very low
complexity using superposition queries.  We fix two arbitrary message
blocks $\alpha_0, \alpha_1$, with $\alpha_0 \ne \alpha_1$, and we
define the following function:
\begin{align*}
  f: \{0,1\} \times \{0,1\}^{n} &\to \{0,1\}^{n} \\
  b, x \quad &\mapsto
    \CBCMAC(\alpha_b \| x) = E_{k'}\left(E_k\big(x \oplus E_k(\alpha_b)\big)\right).
 \end{align*}
The function $f$ can be computed with a single call to the
cryptographic oracle, and we can build a quantum circuit for $f$ given
a black box quantum circuit for $\CBCMAC_k$.  Moreover, 
$f$ satisfies  the promise of Simon's problem with $s = 1\|
E_k(\alpha_0) \oplus E_k(\alpha_1)$:
\begin{align*}
  f(0, x) &= E_{k'}(E_k(x \oplus E_k(\alpha_1))),\\
  f(1, x) &= E_{k'}(E_k(x \oplus E_k(\alpha_0))),\\
  f(b, x) &= f(b \oplus 1, x \oplus E_k(\alpha_0) \oplus E_k(\alpha_1)).
\shortintertext{More precisely:}
  f(b',x') = f(b,x) &\Leftrightarrow x \oplus E_k(\alpha_b) = x' \oplus E_k(\alpha_{b'})\\
&\Leftrightarrow \begin{cases}
      x' \oplus x = 0 & \text{if $b' = b$}\\
      x' \oplus x = E_k(\alpha_0)\oplus E_k(\alpha_1) & \text{if $b' \ne b$}
    \end{cases}
\end{align*}
Therefore, an application of Simon's
algorithm returns $E_k(\alpha_0) \oplus E_k(\alpha_1)$.
This allows to forge messages easily:
\begin{enumerate}
\item Query the tag of $\alpha_0 \| m_1$ for an arbitrary block
  $m_1$;
\item The same tag is  valid for $\alpha_1 \| m_1 \oplus E_k(\alpha_0) \oplus
  E_k(\alpha_1)$.
\end{enumerate}
In order to break the formal notion of EUF-qCMA security, we must
produce $q+1$ valid tags with only $q$ queries to the oracle.  Let
$q' = O(n)$ denote the number of of quantum queries made to learn
$E_k(\alpha_0)\oplus E_k(\alpha_1)$.  
The attacker will repeats the forgery step step $q'+1$ times, in order
to produce $2(q'+1)$ messages with valid tags, after a total of $2q'+1$
classical and quantum queries to the cryptographic oracle. Therefore, $\CBCMAC$
is broken by a quantum existential forgery attack.

\smallskip

After some exchange at early stages of the work, an extension of this forgery attack has been found by Santoli and Schaffner~\cite{SS}.
Its main advantage is to handle oracles that accept input of fixed length,
while our attack works for oracles accepting messages of variable length.

\hide{
\begin{figure}
  \centering
  \begin{tikzpicture}
      \node       (k) at (0,0) {0};
      \node      (alpha0) at (1,1.5) {$\alpha_0$ if $b=0$};
      \node      (alpha1) at (1,1) {$\alpha_1$ if $b=1$};
      \node[xor] (x0) at (1,0) {} edge[<-] (k) edge[<-] (m0);
      \node[fun,minimum size=1cm] (p0) at (2.5,0) {$E_k$} edge[<-] (x0);
      \node      (xx0) at (4,1.5) {$x$};
      \node[xor] (x1) at (4,0) {} edge[<-] (p0) edge[<-] (xx0);
      \node[fun,minimum size=1cm] (p1) at (5.5,0) {$E_k$} edge[<-] (x1);
  \end{tikzpicture}
  
  \caption{Representation of the generic quantum forgery attack on
    CBC-MAC.~\label{fig:CBC-Simon} }
\end{figure}
}

\subsubsection{PMAC.}

PMAC is a parallelizable block-cipher based MAC designed by
Rogway~\cite{DBLP:conf/asiacrypt/Rogaway04}.  PMAC is based on the XE construction: the
construction uses secret offsets $\Delta_i$ derived from the secret
key to turn the block cipher into a tweakable block cipher.  More
precisely, the PMAC
algorithm is defined as
\begin{align*}
  c_i &= E_k(m_i \oplus \Delta_i) &  \PMAC(M) = E^*_k\big(m_{\ell} \oplus \sum c_i\big)
\end{align*}
where $E^*$ is a tweaked variant of $E$. We omit the generation of
the secret offsets because they are irrelevant to our attack.

\subsubsection{First attack.}
When PMAC is used with two-block messages, it has the same structure
as CBC-MAC: $\PMAC(m_1\|m_2) = E^*_k(m_2 \oplus E_k(m_1 \oplus
\Delta_0))$.  Therefore we can use the attack of the previous section
to recover $E_k(\alpha_0) \oplus E_k(\alpha_1)$ for arbitrary values
of $\alpha_0$ and $\alpha_1$.  Again, this leads to
a simple forgery attack.
First, query the tag of $\alpha_0 \| m_1 \| m_2$ for arbitrary blocks
  $m_1$, $m_2$.
The same tag is  valid for $\alpha_1 \| m_1 \| m_2 \oplus E_k(\alpha_0) \oplus
  E_k(\alpha_1)$.
As for $\CBCMAC$, these two steps can be repeated $t+1$ times, where $t$ is the number of quantum
queries issued. The adversary then produces $2(t+1)$ messages after only $2t+1$ queries to the cryptographic oracle.

\subsubsection{Second attack.}

We can also build another forgery attack on PMAC where we recover the
difference between two offsets $\Delta_i$, following the attack
against LRW given in Section~\ref{sec:LRW}.  More precisely, we use
the following function:
\begin{align*}
  f: \{0,1\}^{n} &\to \{0,1\}^{n} \\
  m \quad &\mapsto \PMAC(m \| m \| 0^n) = E_k^*\left(E_k(m \oplus \Delta_0) \oplus E_k(m \oplus \Delta_1)\right).
\end{align*}
In particular, it satisfies $f(m \oplus s) = f(m)$ with $s = \Delta_0 \oplus \Delta_1$.
Furthermore, we can show that $\varepsilon(f,s) \le 1/2$ when $E$
is a good block cipher\footnote{Since this attack is just a special case of the
LRW attack of Section~\ref{sec:LRW}, we don't repeat the detailed
proof.}, and we can apply Simon's algorithm to recover $\Delta_0 \oplus
\Delta_1$.
This allows to create forgeries as follows:
\begin{enumerate}
\item Query the tag of $m_1 \| m_1$ for an arbitrary block
  $m_1$;
\item The same tag is  valid for $m_1 \oplus \Delta_0 \oplus \Delta_1
  \| m_1 \oplus \Delta_0 \oplus \Delta_1$.
\end{enumerate}

As mentioned in Section~\ref{sec:LRW}, the offsets in PMAC are defined
as $\Delta_i = \gamma(i) \cdot L$, with $L = E_k(0)$ and $\gamma$ a Gray
encoding.  This allows to recover $L$ from $\Delta_0 \oplus \Delta_1$, as
$L = (\Delta_0 \oplus \Delta_1) \cdot \left(\gamma(0) \oplus
  \gamma(1)\right)^{-1}$.  Then we can compute all the values
$\Delta_i$, and forge arbitrary messages.

\ifnopromise
\smallskip
We can also mount an attack without any assumption on
$\varepsilon(f,s)$, using Theorem~\ref{th:nopromise}.  Indeed, with a
proper choice of parameters, Simon's algorithm will return a value
$t \ne 0$ that satisfies $\Pr_x[ f(x \oplus t) = f(x)] \ge 1/2$.  This
value is not necessarily equal to $s$, but it can also be used to create
forgeries in the same way, with success probability at least $1/2$.
\fi

\begin{figure}
{
\begin{minipage}[t]{0.4\textwidth}
\centering
	\begin{tikzpicture}
	\node       (b) at (0,2) {$\ket b$};
	\node[fun, minimum size=1cm]		(alphab) at (1,2) {$U_\alpha$} edge[-] (b);
	\node	(binpCBC) at (2.7,2) {} edge[-] (alphab);
	\node       (x) at (0,1) {$\ket x$};
	\node	(xinpCBC) at (2.7,1) {} edge[-] (x);
	\node       (0) at (0,0) {$\ket 0$};
	\node	(xinpCBC) at (2.7,0) {} edge[-] (0);	
	\node[fun, minimum width=3cm, minimum height=1cm, rotate=90] at (2.5,1) {CBC-MAC};
	\node	(boutCBC) at (2.9,2) {};
	\node	(xoutCBC) at (2.9,1) {};
	\node	(x0outCBC) at (2.9,0) {};
	\node[fun, minimum size=1cm]		(alphabm1) at (4,2) {$U_\alpha^{-1}$} edge[-] (boutCBC);
	\node       (b) at (5,2) {$\ket b$} edge[-] (alphabm1);
	\node       (b) at (5,1) {$\ket x$} edge[-] (xoutCBC);
	\node       (b) at (5,0) {$\ket{f(b,x)}$} edge[-] (x0outCBC); 
	\end{tikzpicture}
\caption{Simon's function for CBC-MAC.}
\end{minipage}
} \hfill
{
\begin{minipage}[t]{0.4\textwidth}
	\centering
	\begin{tikzpicture}
	\node[left]       (m)    at (0.5,3) {$\ket m$};
	\node[left]       (0in2) at (0.5,2) {$\ket 0$};
	\node[left]       (0in1) at (0.5,1) {$\ket 0$};
	\node[left]       (0in0) at (0.5,0) {$\ket 0$};
	\node[fork]	(control) at (1,3) {} edge[-] (m);
	\node[xor]		(xor) at (1,2) {} edge[-] (0in2) edge[-] (control);
	\node	(minPM) at (1.7,3) {} edge[-] (control);
	\node	(0inPM2) at (1.7,2) {} edge[-] (xor);
	\node	(0inPM1) at (1.7,1) {} edge[-] (0in1);
	\node	(0inPM0) at (1.7,0) {} edge[-] (0in0);	
	\node[fun, rotate=90, minimum width=4cm, minimum height=1cm] at (2,1.5) {PMAC};
	\node	(moutPM) at (2.4,3) {};
	\node	(0outPM2) at (2.4,2) {};
	\node	(0outPM1) at (2.4,1) {};
	\node	(0outPM0) at (2.4,0) {};
	\node[fork]	(control2) at (3,3) {} edge[-] (moutPM);
	\node[xor]		(xor2) at (3,2) {} edge[-] (0outPM2) edge[-] (control2);
	\node[right]       (mout)  at (3.5,3) {$\ket m$} edge[-] (control2);
	\node[right]       (0out2) at (3.5,2) {$\ket 0$} edge[-] (xor2);
	\node[right]       (out1)  at (3.5,1) {$\ket 0$} edge[-] (0outPM1);
	\node[right]       (out1)  at (3.5,0) {$\ket{f(b,x)}$} edge[-] (0outPM0); 
	\end{tikzpicture}
		\caption{Simon's function for the second attack against PMAC.}\label{pmac2}

	\end{minipage}}
\end{figure}

\subsection{Randomized MAC: GMAC}

GMAC is the underlying MAC of the widely used GCM standard, designed by
McGrew and Viega~\cite{DBLP:conf/indocrypt/McGrewV04}, and standardized
by NIST.  GMAC follows the Carter-Wegman
construction~\cite{DBLP:conf/stoc/CarterW77}: it is built from a
universal hash function, using polynomial evaluation in a Galois field.
As opposed to the constructions of the previous sections, GMAC is a
randomized MAC; it requires a second input $N$, which must be
non-repeating (a nonce).  GMAC is essentially defined as:
\begin{align*}
  \GMAC(N,M) &= \text{GHASH}(M \| \text{len}(M)) \oplus E_k(N||1) \\
  \text{GHASH}(M) &= \sum_{i=1}^{\text{len}(M)} m_i \cdot H^{\text{len}(M)-i+1}  \quad \text{with} \,   H = E_k(0),
\end{align*}
where $\text{len}(M)$ is the length of $M$.

\begin{figure}
  \centering
  \begin{tikzpicture}
      \node       (k) at (0,0) {0};
      \node      (m0) at (1,1) {$m_1$};
      \node[xor] (x0) at (1,0) {} edge[<-] (k) edge[<-] (m0);
      \node[fun,minimum size=1cm] (p0) at (2,0) {$\odot H$} edge[<-] (x0);
      \node      (m1) at (3,1) {$m_2$};
      \node[xor] (x1) at (3,0) {} edge[<-] (p0) edge[<-] (m1);
      \node[fun,minimum size=1cm] (p1) at (4,0) {$\odot H$} edge[<-] (x1);
      \node      (m2) at (5,1) {$\text{len}(M)$};
      \node[xor] (x2) at (5,0) {} edge[<-] (p1) edge[<-] (m2);
      \node[fun,minimum size=1cm] (p2) at (6,0) {$\odot H$} edge[<-] (x2);
      \node[xor] (x) at (7.5,0) {} edge[<-] (p2);
      \node[fun,minimum size=1cm] (e) at (7.5,1) {$E_{k}$} edge[->] (x);
      \node at (7.5,2) {$N\|1$} edge[->] (e);
      \node      (t)  at (9,0) {$\tau$} edge[<-] (x);
  \end{tikzpicture}
  
  \caption{GMAC}
  \label{fig:gmac}
\end{figure}

\subsubsection{Attack.}

When the polynomial is evaluated with Horner's rule, the structure of
GMAC is similar to that of CBC-MAC (see Figure~\ref{fig:gmac}).  For a two-block message,
we have $\GMAC(m_1\|m_2) = \big((m_1 \cdot H) \oplus m_2\big) \cdot H
\oplus E_k(N\|1)$.  Therefore, we us the same
$f$ as in the CBC-MAC attack, with fixed blocks $\alpha_0$ and $\alpha_1$:
\begin{align*}
  f_N: \{0,1\} \times \{0,1\}^{n} &\to \{0,1\}^{n} \\
  b, x  \quad &\mapsto
    \GMAC(N, \alpha_b\|x)= \alpha_b \cdot H^2 \oplus x \cdot H \oplus E_k(N||1).
\end{align*}
In particular, we have:
\begin{align*}
  f(b',x') = f(b,x) &\Leftrightarrow
\alpha_b \cdot H^2 \oplus x \cdot H = \alpha_{b'} \cdot H^2 \oplus x' \cdot H \\
  &\Leftrightarrow \begin{cases}
    x' \oplus x = 0 & \text{if $b' = b$}\\
    x' \oplus x = (\alpha_0 \oplus \alpha_1)\cdot H & \text{if $b' \ne b$}\\
  \end{cases}
\end{align*}
Therefore $f_N$ satisfies the promise of Simon's algorithm with $s = 1 \| (\alpha_0
\oplus \alpha_1) \cdot H$.

\subsubsection{Role of the nonce.}  There is an important caveat
regarding the use of the nonce.
In a classical setting, the nonce is chosen by the adversary under the
constraint that it is non-repeating, \emph{i.e.} the oracle computes $N,
M \mapsto \GMAC(N, M)$.  However, in the quantum setting, we don't have
a clear definition of non-repeating if the nonce can be in
superposition.  To sidestep the issue, we use a weaker security notion
where the nonce is chosen at random by the oracle, rather than by the
adversary (following the IND-qCPA definition of
\cite{DBLP:conf/crypto/BonehZ13}).  The oracle is then $M \mapsto (r,
\GMAC(r,M))$.  If we can break the scheme in this model, the attack will
also be valid with any reasonable CPA security definition.

In this setting we can access the function $f_N$ only for a random value
of $N$.  In particular, we cannot apply Simon's algorithm as is, because
this requires $O(n)$ queries to the \emph{same} function $f_N$.
However, a single step of Simon's algorithm requires a single query to
the $f_N$ function, and returns a vector orthogonal to $s$, for any
random choice of $N$.  Therefore, we can recover $(\alpha_0 \oplus \alpha_1)
\cdot H$ after $O(n)$ steps, even if each step uses a different value of
$N$.  Then, we can recover $H$ easily, and it is easy to generate
forgeries when $H$ is known:
\begin{enumerate}
\item Query the tag of $N, m_1 \| m_2$ for arbitrary blocks
  $m_1$, $m_2$ (under a random nonce $N$).
\item The same tag is valid for $m_1 \oplus 1 \| m_2 \oplus H$ (with the
  same nonce $N$).
\end{enumerate}
As for $\CBCMAC$, repeating these two steps leads to an existential forgery attack.
\subsection{Classical Authenticated Encryption Schemes: GCM and OCB}

We now give applications of Simon's algorithm to break the security of
standardized authenticated encryption modes.  The attacks are similar to the
attacks against authentication modes, but these authenticated encryption
modes are nonce-based.  Therefore we have to pay special attention to
the nonce, as in the attack against GMAC.  In the following, we assume
that the nonce is randomly chosen by the MAC oracle, in order to avoid
issues with the definition of non-repeating nonce in a quantum setting.

\subsubsection{Extending MAC attacks to authenticated encryption schemes.}

We first present a generic way to apply MAC attacks in the context of an
authenticated encryption scheme.  More precisely, we assume that the tag
of the authenticated encryption scheme is computed as $f(g(A), h(M,N))$,
\emph{i.e.} the authentication of the associated data $A$ is independent
of the nonce $N$.  This is the case in many practical schemes
(\emph{e.g.} GCM, OCB) for efficiency reasons.

In this setting, we can use a technique similar to our attack against
GMAC: we define a function $M \mapsto f_N(M)$ for a fixed nonce $N$,
such that for any nonce $N$, $f_N(M) = f_N(M \oplus \Delta)$ for some
secret value $\Delta$.  Next we use Simon's algorithm to recover
$\Delta$, where each step of Simon's algorithm is run with a
random nonce, and returns a vector orthogonal to $\Delta$.  Finally, we
can recover $\Delta$, and if $f_N$ was carefully built, the knowledge of
$\Delta$ is sufficient for a forgery attack.

The CCM mode is a notable exception, where all the
computations depend on the nonce.  In particular, there is no obvious
way to apply our attacks to CCM.

\subsubsection{Extending GMAC attack to GCM.}
GCM is one of the most widely used authenticated encryption modes, designed
by McGrew and Viega~\cite{DBLP:conf/indocrypt/McGrewV04}.  GMAC is the
composition of the counter mode for encryption with GMAC (computed over
the associated data and the ciphertext) for authentication.

In particular, when the message is empty, GCM is just GMAC, and we can
use the attack of the previous section to recover the hash
key $H$.  This immediately allows a forgery attack.

\subsubsection{OCB.}
OCB is another popular authenticated encryption mode, with a very high
efficiency, designed by Rogaway \emph{et al.}~\cite{DBLP:conf/ccs/RogawayBBK01,DBLP:conf/asiacrypt/Rogaway04,DBLP:conf/fse/KrovetzR11}.  Indeed, OCB requires only $\ell$ block cipher calls to
process an $\ell$-block message, while GCM requires $\ell$ block cipher
calls, and $\ell$ finite field operations.  OCB is build from the LRW
construction discussed in Section~\ref{sec:LRW}.  OCB takes as input a
nonce $N$, a message $M = m_1 \| \ldots \| m_\ell$, and associated data
$A = a_1 \| \ldots a_{\text{\emph{@}}}$, and returns a ciphertext $C =
c_1 \| \ldots \| c_\ell$ and a tag $\tau$:
\begin{align*}
  c_i = E_k(m_i \oplus \Delta^N_i) \oplus \Delta^N_i, \quad 
  \tau = E_k\left(\Delta'^N_\ell \oplus \sum m_i\right) \oplus \sum b_i, \quad 
  b_i = E_k(a_i \oplus \Delta_i).
\end{align*}

\subsubsection{Extending PMAC attack to OCB.}
In particular, when the message is empty, OCB reduces to a randomized
variant of PMAC:
\begin{align*}
  \OCB_k(N, \varepsilon, A) &= \phi_k(N) \oplus \sum b_i, &
  b_i &= E_k(a_i \oplus \Delta_i).
\end{align*}
Note that the $\Delta_i$ values used for the associated data are
independent of the nonce $N$.
Therefore, we can apply the second PMAC attack previously given, using
the following function:
\begin{align*}
  f_N: \{0,1\}^{n} &\to \{0,1\}^{n} \\
  x \quad &\mapsto \OCB_k(N, \varepsilon, x \| x) \\
  f_N(x)    &=E_k(x \oplus \Delta_0) \oplus E_k(x \oplus \Delta_1) \oplus \phi_k(N)
\end{align*}
Again, this is a special case of the LRW attack of
Section~\ref{sec:LRW}.  The family of functions satisfies
$f_N(a \oplus \Delta_0 \oplus \Delta_1) = f_N(a)$, for any $N$, and
$\varepsilon(f_N,\Delta_0 \oplus \Delta_1) \le 1/2$ with overwhelming
probability if $E$ is a PRP.  Therefore we can use the variant of
Simon's algorithm to recover $\Delta_0 \oplus \Delta_1$.  Two messages with valid tags
can then be generated by a single classical queries:
\begin{enumerate}
\item Query the authenticated encryption $C,\tau$ of $M, a \| a$ for an arbitrary
  message $M$, and an arbitrary block $a$ (under a random nonce $N$).
\item $C,\tau$ is also a valid authenticated encryption of $M, a \oplus \Delta_0 \oplus \Delta_1
  \| a \oplus \Delta_0 \oplus \Delta_1$, with the same nonce $N$.
\end{enumerate}
Repeating these steps lead again to an existential forgery attack.
\subsubsection{Alternative attack against OCB.}

For some versions of OCB, we can also mount a different attack targeting the encryption part
rather than the authentication part.  The goal of this attack is also to
recover the secret offsets, but we target the $\Delta^N_i$ used for the
encryption of the message.
More precisely, we use the following function:
\begin{align*}
  f_i: \{0,1\}^{n} &\to \{0,1\}^{n} \\
  m \quad &\mapsto c_1 \oplus c_2, \text{where $(c_1, c_2, \tau) = \OCB_k(N, m \| m, \varepsilon)$} \\
  f_i(m)    &=E_k(m \oplus \Delta^N_{1}) \oplus \Delta^N_{1} \oplus E_k(m \oplus \Delta^N_{2}) \oplus \Delta^N_{2}
\end{align*}
This function satisfies $f_N(m \oplus \Delta^N_1 \oplus \Delta^N_2) =
f_N(m)$ and $\varepsilon(f_N,\Delta^N_0 \oplus \Delta^N_1) \le 1/2$, with
the same arguments as previously.
Moreover, in OCB1 and OCB3, the offsets are derived as $\Delta^N_i =
\Phi_k(N) \oplus \gamma(i) \cdot E_k(0)$ for some function $\Phi$ (based
on the block cipher $E_k$).  In particular,
$\Delta^N_1 \oplus \Delta^N_2$ is independent of $N$:
\[
\Delta^N_1 \oplus \Delta^N_2 = (\gamma(1) \oplus \gamma(2)) \cdot E_k(0).
\]
Therefore, we can apply Simon's algorithm to recover $\Delta^N_1 \oplus \Delta^N_2$.
Again, this leads to a forgery attack, by repeating the following two steps:
\begin{enumerate}
\item Query the authenticated encryption $c_1\|c_2,\tau$ of $m\|m, A$
  for an arbitrary block $m$, and arbitrary associated data
  $A$ (under a random nonce $N$).
\item $c_2 \oplus \Delta^N_0 \oplus \Delta^N_1 \| c_1 \oplus \Delta^N_0 \oplus
  \Delta^N_1, \tau$ is also a valid authenticated encryption of $m \oplus \Delta^N_0 \oplus \Delta^N_1
  \| m \oplus \Delta^N_0 \oplus \Delta^N_1, A$ with the same nonce $N$.
\end{enumerate}
The forgery is valid because we swap the inputs of the first and second
block ciphers.  In addition, we have 
$\sum m_i = \sum m'_i$, so that the tag is still valid.

\subsection{New Authenticated Encryption Schemes: CAESAR Candidates}

In this section, we consider recent proposals for authenticated
encryption, submitted to the ongoing CAESAR competition. Secret key
cryptography has a long tradition of competitions: AES and \mbox{SHA-3}
for example, were chosen after the NIST competitions organized in 1997
and 2007, respectively.  The CAESAR
competition\footnote{\url{http://competitions.cr.yp.to/}} aims at stimulating
research on authenticated encryption schemes, and to define a portfolio
of new authenticated encryption schemes.  The competition is currently
in the second round, with 29 remaining algorithms.

First, we point out that the attacks of the previous sections can be used
to break several CAESAR candidates:
\begin{itemize}
\item CLOC~\cite{DBLP:conf/fse/IwataM0M14} uses CBC-MAC to authenticate the message, and the associated
  data is processed independently of the nonce.  Therefore, the CBC-MAC
  attack can be extended to CLOC\footnote{This is not the case for the
    related mode SILC, because the nonce is processed before the data in
    CBC-MAC.}.
\item AEZ~\cite{DBLP:conf/eurocrypt/HoangKR15}, COPA~\cite{DBLP:conf/asiacrypt/AndreevaBLMTY13}, OTR~\cite{DBLP:conf/eurocrypt/Minematsu14} and POET~\cite{DBLP:conf/fse/AbedFFLLMW14} use a variant of PMAC to authenticate the associated
  data.  In both cases, the nonce is not used to process the associated
  data, so that we can extend the PMAC attack as we did against
  OCB\footnote{Note that AEZ, COPA and POET also claim security when the nonce
    is misused, but our attacks are nonce-respecting.}.
\item The authentication of associated data in OMD~\cite{DBLP:conf/sacrypt/CoglianiMNCRVV14} and Minalpher~\cite{CAESAR_Minalpher} are
  also variants of PMAC (with a PRF that is not block cipher), and the
  attack can be applied.
\end{itemize}

In the next section, we show how to adapt the PMAC attack to Minalpher
and OMD, since the primitives are different.

\subsubsection{Minalpher.}

Minalpher~\cite{CAESAR_Minalpher} is a permutation-based CAESAR candidate, where the permutation
is used to build a tweakable block-cipher using the tweakable
Even-Mansour construction.  When the message is empty (or fixed), the
authentication part of Minalpher is very similar to PMAC.  With
associated data $A = a_1 \| \ldots a_{\text{\emph{@}}}$, the tag is
computed as:
\begin{align*}
  b_i &= P(a_i \oplus \Delta_i) \oplus \Delta_i &
  \tau &= \phi_k\left(N, M, a_{\text{\emph{@}}} \oplus \sum_{i=1}^{\text{\emph{@}}-1} b_i\right)\\
  \Delta_i &= y^i \cdot L' &
  L' &= P(k \| 0) \oplus (k \| 0)
\end{align*}
where $\phi_k$ is a permutation (we omit the description of $\phi_k$ because it is irrelevant for our
attack).  Since the tag is a function of $a_{\text{\emph{@}}} \oplus
\sum_{i=1}^{\text{\emph{@}}-1} b_i$, we can use the same attacks as
against PMAC.  For instance, we define the following function:
\begin{align*}
  f_N: \{0,1\} \times \{0,1\}^{n} &\to \{0,1\}^{n} \\
  b, x \quad &\mapsto \text{Minalpher}(N, \varepsilon, \alpha_b \| x)= \phi_k(N, \varepsilon, P(\alpha_b \oplus \Delta_1) \oplus \Delta_1 \oplus x).
\end{align*}
In particular, we have:
\begin{align*}
  f_N(b',x') = f_N(b,x) &\Leftrightarrow
P(\alpha_{b'} \oplus \Delta_1) \oplus x' = P(\alpha_b \oplus \Delta_1) \oplus x \\
&\Leftrightarrow \begin{cases}
      x' \oplus x = 0 & \text{if $b' = b$}\\
      x' \oplus x = P(\alpha_0\oplus\Delta_1)\oplus P(\alpha_1\oplus\Delta_1) & \text{if $b' \ne b$}\\
    \end{cases}
\end{align*}
Since $s=P(\alpha_0\oplus\Delta_1)\oplus P(\alpha_1\oplus\Delta_1)$ is
independent of $N$, we can easily apply Simon's algorithm to
recover $s$, and generate forgeries.

\subsubsection{OMD.}

OMD~\cite{DBLP:conf/sacrypt/CoglianiMNCRVV14} is a compression-function-based CAESAR candidate.  The internal
primitive is a keyed compression function denoted $F_k$.  Again, when
the message is empty the authentication is very similar to PMAC.  With
associated data $A = a_1 \| \ldots a_{\text{\emph{@}}}$, the tag is
computed as:
\begin{align*}
  b_i &= F_k(a_i \oplus \Delta_i) &
  \tau &= \phi_k(N, M) \oplus \sum b_i
\end{align*}
We note that the $\Delta_i$ used for the associated data do not depend
on the nonce.  Therefore we can use the second PMAC attack with the
following function:
\begin{align*}
  f_N: \{0,1\}^{n} &\to \{0,1\}^{n} \\
  x \quad &\mapsto \OMD(N, \varepsilon, x \| x)\\
  f_N(x) &= \phi_k(N, \varepsilon) \oplus F_k(x \oplus \Delta_1) \oplus F_k(x
  \oplus \Delta_2)
\end{align*}
This is the same form as seen when extending the PMAC attack to OCB,
therefore we can apply the same attack to recover $s= \Delta_1 \oplus
\Delta_2$ and generate forgeries.

\section{Simon's algorithm applied to slide attacks}
\label{sec:slide}
In this section we show how Simon's algorithm can be applied to a cryptanalysis family: slide attacks. 
In this case, the complexity of the attack drops again exponentially, from $O(2^{n/2})$ to $O(n)$ and therefore becomes much more dangerous. To the best of our knowledge this is the first symmetric cryptanalytic technique that has an exponential speed-up in the post-quantum world.

\subsubsection{The principle of slide attacks}
\label{treyfer}
In 1999, Wagner and Biryukov introduced the technique called \emph{slide attack}~\cite{DBLP:conf/fse/BiryukovW99}. It can be applied to block ciphers 
made of $r$ applications of an identical round function $R$, each one parametrized by the same key $K$.
The attack works independently of the number of rounds, $r$.
Intuitively, for the attack to work, $R$ has to be vulnerable to known plaintext attacks.

The attacker collects $2^{n/2}$ encryptions of plaintexts. Amongst these couples of plaintext-ciphertext, 
with large probability, he gets a ``slid" pair, that is, a pair of couples $(P_0, C_0)$ and $(P_1, C_1)$ 
such that $R(P_0)=P_1$.
This immediately implies that $R(C_0)=C_1$. For the attack to work, the function $R$ needs to allow for an efficient recognition of such pairs, which in turns makes
the key extraction from~$R$ easy. A trivial application of this attack 
is the key-alternate cipher with blocks of $n$ bits, identical subkeys and no round constants.   
The complexity is then approximately $2^{n/2}$. The speed-up over exhaustive search given by this attack is
then quadratic, similar to the quantum attack based on Grover's algorithm.

This attack is successful, for example, to break the TREYFER block cipher~\cite{DBLP:conf/fse/Yuval97}, with a data complexity of $2^{32}$ and a time complexity of $2^{32+12}=2^{44}$ (where $2^{12}$ is the cost of identifying the slid pair by performing some key guesses). Comparatively,  the cost for an exhaustive search of the key is~$2^{64}$.

\subsubsection{Exponential quantum speed-up of slide attacks}
We consider the attack represented in Figure~\ref{fig:Slide-Simon}. The unkeyed round function is denoted $P$ and the whole encryption function $E_k$.

\begin{figure}
  \centering
  \begin{tikzpicture}
      \node       (P0) at (-0.5,0) {$P_0$};
      \node      (K1) at (0.5,1) {$K$};
      \node[xor] (x0) at (0.5,0) {} edge[<-] (P0) edge[<-] (K1);
      \node[fun,minimum size=1cm] (p0) at (1.8,0) {$P$} edge[<-] (x0);
      \node[fork]	(fB) at (1,0) {};
      \node	(B) at (1,-0.5) {$B$} edge[<-] (fB);
      \node      (k2) at (3,1) {$K$};
      \node[fork]	(fP1) at (2.6,0) {};
      \node	(vP1) at (2.6,-0.5) {$P_1$} edge[<-] (fP1);
      \node[xor] (x1) at (3,0) {} edge[<-] (p0) edge[<-] (k2);
      \node[fun,minimum size=1cm] (p1) at (4,0) {$P$} edge[<-] (x1);
	\node	(pp) at (5.5,0) {$\ldots$} edge[<-] (p1);
      \node      (k3) at (7,1) {$K$};
      \node[xor] (x2) at (7,0) {} edge[<-] (k3) edge[<-] (pp);
      \node[fun,minimum size=1cm] (p3) at (8,0) {$P$} edge[<-] (x2);
            \node      (k4) at (9,1) {$K$};
      \node[xor] (x3) at (9,0) {} edge[<-] (k4) edge[<-] (p3);      
      \node      (C0)  at (10,0) {$C_0$} edge[<-] (p3);

      \node       (P1) at (2,-2) {$P_1$};
      \node      (K11) at (3,-1) {$K$};
      \node[xor] (x10) at (3,-2) {} edge[<-] (P1) edge[<-] (K11);
      \node[fun,minimum size=1cm] (p10) at (4,-2) {$P$} edge[<-] (x10);
      \node      (k12) at (5,-1) {$K$};
      \node[xor] (x11) at (5,-2) {} edge[<-] (p10) edge[<-] (k12);
      \node[fun,minimum size=1cm] (p11) at (6,-2) {$P$} edge[<-] (x11);
	\node	(pp1) at (7.5,-2) {$\ldots$} edge[<-] (p11);
      \node      (k13) at (9,-1) {$K$};
      \node[xor] (x12) at (9,-2) {} edge[<-] (k13) edge[<-] (pp1);
      \node[fork] 	(fC0) at (9.4,-2) {};
      \node	(vC0) at (9.4,-1.5) {$C_0$} edge[<-] (fC0);
      \node[fun,minimum size=1cm] (p13) at (10.2,-2) {$P$} edge[<-] (x12);
      \node[fork] 	(fA) at (11,-2) {};
      \node	(A) at (11,-1.5) {$A$} edge[<-] (fA);
      \node      (k14) at (11.3,-1) {$K$};
      \node[xor] (x13) at (11.3,-2) {} edge[<-] (k14) edge[<-] (p13);      
      \node      (C1)  at (12,-2) {$C_1$} edge[<-] (p13);

  \end{tikzpicture}
  
  \caption{Representation of a slid-pair used in a slide attack.~\label{fig:Slide-Simon}}
\end{figure}
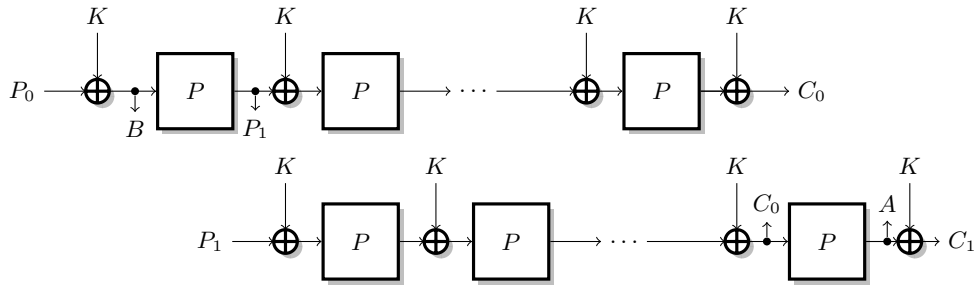
We define the following function:
\begin{align*}
  f: \{0,1\} \times \{0,1\}^{n} &\to \{0,1\}^{n} \\
  b,x \quad &\mapsto
  \begin{cases}
    P(E_k(x)) \oplus x & \text{if $b = 0$,}\\
    E_k(P(x)) \oplus x & \text{if $b = 1$.}
  \end{cases}
\end{align*}
The slide property shows that all $x$ satisfy $P(E_k(x)) \oplus k =
E_k(P(x  \oplus k))$.  This implies that $f$ satisfies the promise of
Simon's problem with $s = 1\|k$:
\begin{align*}
  f(0,x) = P(E_k(x)) \oplus x &= E_k(P(x \oplus k)) \oplus k \oplus x = f(1,x \oplus k).
\end{align*}
In order to apply Theorem~\ref{th:approx}, we bound
$\varepsilon(f,1\|k)$, assuming that both $E_k \circ P$ and $P \circ E_k$ are indistinguishable from random permutations.
If $\varepsilon(f,1\|k) > 1/2$,
there exists $(\tau,t)$ with $(\tau,t) \not\in \{(0,0), (1,k)\}$ such that:
  $\mathrm{Pr}[f(b,x) = f(b\oplus \tau, x\oplus t)] > 1/2$.
Let us assume $\tau = 0$.  This implies
\begin{align*}
    \mathrm{Pr}[f(0,x) = f(0, x\oplus t)] > 1/2  \quad \text{or} \quad 
    \mathrm{Pr}[f(1,x) = f(1, x\oplus t)] > 1/2,
\end{align*}
 which is equivalent to 
 \begin{align*}
  \mathrm{Pr}[P(E_k(x)) = P(E_k(x \oplus t)) \oplus t]
> 1/2 \quad  \text{or} \quad 
  \mathrm{Pr}[E_k(P(x)) = E_k(P(x \oplus t)) \oplus t]
> 1/2.
\end{align*}

In particular, there is a differential in $P \circ E_k$ or $E_k \circ P$
with probability $1/2$.

Otherwise, $\tau = 1$.  This implies
\begin{align*}
&  \mathrm{Pr}[P(E_k(x)) \oplus x = E_k(P(x \oplus t)) \oplus x \oplus t]
> 1/2\\
\text{\emph{i.e.}} \quad &  \mathrm{Pr}[E_k(P(x \oplus k)) \oplus k = E_k(P(x \oplus t)) \oplus t]
> 1/2.
\end{align*}
Again, it means there is a differential in $E_k \circ P$ with probability $1/2$.

Finally we conclude that $\varepsilon(f,1\|k) \le 1/2$, unless
$E_k \circ P$ or $P \circ E_k$ have differentials with probability
$1/2$.  If $E_k$ behave as a random permutation, $E_k \circ P$ and $P
\circ E_k$ also behave as random permutations, and these differential
are only found with negligible probability.
Therefore, we can apply Simon's algorithm, following
Theorem~\ref{th:approx}, and recover $k$.

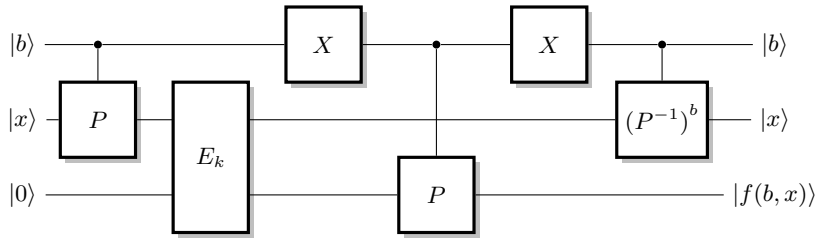
\begin{figure}
	\centering
	\begin{tikzpicture}
	\node       (binp) at (0,2) {$\ket b$};
	\node	(xinp) at (0,1) {$\ket x$};
	\node	(0inp) at (0,0) {$\ket 0$};
	\node[fork]	(control) at (1,2) {} edge[-] (binp);
	\node[fun, minimum size=1cm]		(F1) at (1,1) {$P$} edge[-] (xinp) edge[-] (control);
	\node	(xinEk) at (2.1,1) {} edge[-] (F1);
	\node	(0inEk) at (2.1,0) {} edge[-] (0inp);
	\node[fun, minimum width=1cm, minimum height=2cm]		(Ek) at (2.5,0.5) {$E_k$};
	\node	(xoutEk) at (2.9,1) {};
	\node	(0outEk) at (2.9,0) {};
	\node[fun, minimum size=1cm]		(X1) at (4,2) {$X$} edge[-] (control);
	\node[fork]	(control2) at (5.5,2) {} edge[-] (X1);
	\node[fun, minimum size=1cm]		(F2) at (5.5,0) {$P$} edge[-] (0outEk) edge[-] (control2);
	\node[fun, minimum size=1cm]		(X2) at (7,2) {$X$} edge[-] (control2);
	\node[fork]	(control3) at (8.5,2) {} edge[-] (X2);
	\node[fun, minimum size=1cm]		(F3) at (8.5,1) {${(P^{-1})}^b$} edge[-] (xoutEk) edge[-] (control3);
	\node       (bout) at (10,2) {$\ket b$} edge[-] (control3);
	\node	(xout) at (10,1) {$\ket x$} edge[-] (F3);
	\node	(0out) at (10,0) {$\ket {f(b,x)}$} edge[-] (F2);	
	\end{tikzpicture}
\caption{Simon's function for slide attacks. The $X$ gate is the quantum equivalent of the \texttt{NOT} gate that flips the qubit $\ket 0$ and $\ket 1$.}
\end{figure}

\section{Conclusion}
\label{sec:conc}
We have been able to show that symmetric cryptography is far from ready for the post quantum world. We have found exponential speed-ups on attacks on symmetric cryptosystems.
In consequence, some cryptosystems that are believed to be safe in a classical world become
vulnerable in a quantum world.

With the speed-up on slide attacks, we provided the first known exponential quantum speed-up of a classical attack. This attack now becomes very powerful.  An interesting follow-up would be to seek other such speed-ups of generic techniques. 
For authenticated encryption, we have shown that many modes of operations that are believed to be solid and secure in the classical world, become completely broken in the post-quantum world. More constructions might be broken following the same ideas.

\section*{Acknowledgements}
We would like to thank Thomas Santoli and Christian Schaffner for sharing an early stage manuscript of their work~\cite{SS},
Michele Mosca for discussions and LTCI for hospitality.
This work was supported by the Commission of the European Communities
through the Horizon 2020 program under project number 645622 PQCRYPTO.
MK acknowledges funding through grants 
ANR-12-PDOC-0022-01 and ESPRC EP/N003829/1.

\ifnopromise\else
\newpage
\fi

\appendix

\section{Proof of Theorem~\ref{th:approx}}
\label{app:proof}

The proof of Theorem \ref{th:approx} is based of the following lemma.
\begin{lemma}\label{lem:sum}
For $t \in \{0,1\}^n$, consider the function $g(x) := 2^{-n} \sum_{y\in t^{\perp}} (-1)^{x \cdot y}$,
where $t^\perp = \{ y\in\{0,1\}^n \, \mathrm{s.t.}\, y \cdot t = 0\}$. for any $x$, it satisfies
\begin{align} 
g(x)= \frac{1}{2}(\delta_{x,0} + \delta_{x,t}).
\end{align}
\end{lemma}

\begin{proof}

If $t=0$ then $g(x) = \sum_{y \in \{0,1\}^n} (-1)^{x \cdot y} =\delta(x,0)$, which proves the claim. 
From now on, assume that $t \ne 0$. 
It is straightforward to check that
$g(0) = g(t) = \frac{1}{2}$
because all the terms of the sum are equal to 1 and there are $2^{n-1}$ vectors $y$ orthogonal to $t$. 
Since $\sum_{x \in \{0,1\}^n} g(x) = 1$, it is sufficient to prove that $g(x) \geq 0$ to establish the claim in the case $t\ne 0$. 
For this,  decompose $g(x)$ into two terms:
\begin{align*}
g(x) = \sum_{y \in E_0} (-1)^{x \cdot y} - \sum_{y \in E_1} (-1)^{x \cdot y} = | E_0| - | E_1|,
\end{align*}
where
$ E_i := \{ y \in \{0,1\}^n \: \mathrm{s.t.} \: y \cdot x= i \: \mathrm{and} \: y \cdot y = 0\}$ for $i=0,1$.
Simple counting shows that:
\begin{align*}
| E_0 | = \left\{ \begin{array}{cl}
2^{n-1} & \mathrm{if} \,  x=0,\\
2^{n-1} & \mathrm{if} \, x=t,\\
2^{n-2} & \mathrm{otherwise}.
\end{array}
\right.
\end{align*}
In particular, $|E_0| \geq |E_1|$ which implies that $g(x) \geq 0$.
\end{proof}

We are now ready to prove Theorem~\ref{th:approx}.
Each call to the main subroutine of Simon's algorithm will return a vector $u_i$. If $cn$ calls are made, one obtains $cn$ vectors $u_1, \ldots, u_{cn}$. 
By construction, $f$ is such that $f(x) = f(x \oplus s)$ and consequently, the $cn$ vectors $u_1, \ldots, u_{cn}$ are all orthogonal to $s$. 
The algorithm is successful provided one can recover the value of $s$ unambiguously, which is the case if the $cn$ vectors span the $(n-1)$-dimensional space orthogonal to $s$. (Let us note that if the space is $(n-d)$-dimensional for some constant $d$, one can still recover $s$ efficiently by testing all the vectors orthogonal to the subspace.)
In other words, the failure probability $p_\text{fail}$ is
\begin{align*}
p_{\text{fail}} &= \mathrm{Pr}[\text{dim}\big(\text{Span}(u_1, \ldots, u_n)\big) \leq n-2]\\
&\leq \mathrm{Pr}[\exists t \in \{0,1\}^n\setminus\{0,s\} \,\text{s.t.} \, u_1 \cdot t = u_2 \cdot t = \cdots = u_{cn} \cdot t =0 ]\\
& \leq \sum_{t \in \{0,1\}^{n} \setminus\{0,s\} } \mathrm{Pr}[ u_1 \cdot t = u_2 \cdot t = \cdots = u_{cn} \cdot t =0 ]\\
&\leq \sum_{t \in \{0,1\}^{n} \setminus\{0,s\} } \big( \mathrm{Pr}[ u_1 \cdot t =0]\big)^{cn} \\
&\leq \max_{t \in \{0,1\}^{n} \setminus\{0,s\} } \big(2 \mathrm{Pr}[ u_1 \cdot t =0]^c \big)^n 
\end{align*}
where the second inequality results from the union bound and the third inequality follows from the fact that the results of the $cn$ subroutines are independent. 

In order to establish the theorem, it is now sufficient to show that $\mathrm{Pr}[ u \cdot t =0]$ is bounded away from $1$ for all $t$, where $u$ is the vector corresponding to the output of Simon's subroutine. We will prove that for all $t\in \{0,1\}^{n} \setminus\{0,s\}$, the following inequality holds:
\begin{align}\label{th-condition}
\mathrm{Pr}_u [ u \cdot t =0] = \frac{1}{2}\big(1 + \mathrm{Pr}_x[f(x) = f(x\oplus t)]\big) \leq \frac{1}{2}(1+\varepsilon(f,s)) \leq \frac{1}{2} (1+p_0).
\end{align}

In Simon's algorithm, one can wait until the last step before measuring both registers. 
The final state before measurement can be decomposed as:
\begin{align*}
2^{-n} \sum_{x \in \{0,1\}^n}\sum_{y \in \{0,1\}^n} (-1)^{x \cdot y} |y\rangle |f(x)\rangle = &2^{-n} \sum_{\substack{y \in \{0,1\}^n\\ \text{s.t.}\, y\cdot t=0} }\sum_{x \in \{0,1\}^n} (-1)^{x \cdot y} |y\rangle |f(x)\rangle \\
& + 2^{-n} \sum_{\substack{y \in \{0,1\}^n\\ \text{s.t.}\, y\cdot t=1} } \sum_{x \in \{0,1\}^n} (-1)^{x \cdot y} |y\rangle |f(x)\rangle.
\end{align*}
The probability of obtaining $u$ such that $u \cdot t = 0$ is given by
\begin{align}
\mathrm{Pr}_u [u \cdot t=0] &= \norm[\Big]{2^{-n} \sum_{\substack{y \in \{0,1\}^n\\ \text{s.t.}\, y\cdot t=0} } |y\rangle \sum_{x \in \{0,1\}^n} (-1)^{x \cdot y} |f(x)\rangle}^2 \nonumber\\
 &=2^{-2n} \sum_{\substack{y \in \{0,1\}^n\\ \text{s.t.}\, y\cdot t=0} }  \sum_{x,x' \in \{0,1\}^n} (-1)^{(x\oplus x') \cdot y}  \langle f(x')|f(x)\rangle \nonumber\\
 &=2^{-2n}\sum_{x,x' \in \{0,1\}^n} \langle f(x')|f(x)\rangle \sum_{\substack{y \in \{0,1\}^n\\ \text{s.t.}\, y\cdot t=0} }   (-1)^{(x \oplus x')\cdot y}  \nonumber\\
 &=2^{-2n}\sum_{x,x' \in \{0,1\}^n} \langle f(x')|f(x)\rangle 2^{n-1} (\delta_{x,x'} + \delta_{x',x \oplus t}) \label{eq:sum} \\
 &= 2^{-(n+1)} \left[  \sum_{x \in \{0,1\}^n} \langle f(x)|f(x) \rangle +  \sum_{x \in \{0,1\}^n} \langle f(x\oplus t)|f(x) \rangle \right] \\
 &= \frac{1}{2} \left[ 1 + \mathrm{Pr}_x [f(x) = f(x\oplus t) \right] 
 \end{align}
where we used Lemma \ref{lem:sum} proven in the appendix in Eq.~\ref{eq:sum}, and $\delta_{x,x'}=1$ if $x=x'$ and 0 otherwise.

\ifnopromise
\section{Proof of Theorem \ref{th:nopromise}}
\label{sec:proof-theorem-2}

Let $t$ be a fixed value and $p_t = Pr_x[f(x \oplus t = f(t)]$.  Following
the previous analysis, the
probability that the $cn$ vectors $u_i$ are orthogonal to $t$ can be
written as $\Pr[ u_1 \cdot t = u_2 \cdot t = \cdots = u_{cn} \cdot t =0
] = \left(\frac{1+p_t}{2}\right)^{cn}$.

In particular, we can bound the probability that Simon's algorithm
returns a value $t$ with $p_t < p_0$:
\[
  \Pr[p_t < p_0]
  = \sum_{t:\, p_t < p_0} \left(\frac{1+p_t}{2}\right)^{cn}
  \le 2^n \times \left(\frac{1+p_0}{2}\right)^{cn}
\]
\fi


\begin{thebibliography}{10}
\providecommand{\url}[1]{\texttt{#1}}
\providecommand{\urlprefix}{URL }

\bibitem{DBLP:conf/fse/AbedFFLLMW14}
Abed, F., Fluhrer, S.R., Forler, C., List, E., Lucks, S., McGrew, D.A., Wenzel,
  J.: Pipelineable on-line encryption. In: Cid, C., Rechberger, C. (eds.) Fast
  Software Encryption - 21st International Workshop, {FSE} 2014, London, UK,
  March 3-5, 2014. Revised Selected Papers. Lecture Notes in Computer Science,
  vol. 8540, pp. 205--223. Springer (2014)

\bibitem{alagic2016computational}
Alagic, G., Broadbent, A., Fefferman, B., Gagliardoni, T., Schaffner, C.,
  Jules, M.S.: Computational security of quantum encryption. arXiv preprint
  arXiv:1602.01441  (2016)

\bibitem{DBLP:conf/pqcrypto/AnandTTU16}
Anand, M.V., Targhi, E.E., Tabia, G.N., Unruh, D.: Post-quantum security of the
  {CBC}, {CFB}, {OFB}, {CTR}, and {XTS} modes of operation. In: Takagi, T.
  (ed.) Post-Quantum Cryptography - 7th International Workshop, PQCrypto 2016,
  Fukuoka, Japan, February 24-26, 2016, Proceedings. Lecture Notes in Computer
  Science, vol. 9606, pp. 44--63. Springer (2016)

\bibitem{DBLP:conf/asiacrypt/AndreevaBLMTY13}
Andreeva, E., Bogdanov, A., Luykx, A., Mennink, B., Tischhauser, E., Yasuda,
  K.: Parallelizable and authenticated online ciphers. In: Sako, K., Sarkar, P.
  (eds.) Advances in Cryptology - {ASIACRYPT} 2013 - 19th International
  Conference on the Theory and Application of Cryptology and Information
  Security, Bengaluru, India, December 1-5, 2013, Proceedings, Part {I}.
  Lecture Notes in Computer Science, vol. 8269, pp. 424--443. Springer (2013)

\bibitem{DBLP:journals/jcss/BellareKR00}
Bellare, M., Kilian, J., Rogaway, P.: The security of the cipher block chaining
  message authentication code. J. Comput. Syst. Sci.  61(3),  362--399 (2000)

\bibitem{bernstein2009introduction}
Bernstein, D.J.: Introduction to post-quantum cryptography. In: Post-quantum
  cryptography, pp. 1--14. Springer (2009)

\bibitem{DBLP:conf/fse/BiryukovW99}
Biryukov, A., Wagner, D.: Slide attacks. In: Knudsen, L.R. (ed.) Fast Software
  Encryption, 6th International Workshop, {FSE} '99, Rome, Italy, March 24-26,
  1999, Proceedings. Lecture Notes in Computer Science, vol. 1636, pp.
  245--259. Springer (1999)

\bibitem{DBLP:conf/eurocrypt/BiryukovW00}
Biryukov, A., Wagner, D.: Advanced slide attacks. In: Preneel, B. (ed.)
  Advances in Cryptology - {EUROCRYPT} 2000, International Conference on the
  Theory and Application of Cryptographic Techniques, Bruges, Belgium, May
  14-18, 2000, Proceeding. Lecture Notes in Computer Science, vol. 1807, pp.
  589--606. Springer (2000)

\bibitem{DBLP:conf/crypto/BlackR00}
Black, J., Rogaway, P.: {CBC} macs for arbitrary-length messages: The three-key
  constructions. In: Bellare, M. (ed.) Advances in Cryptology - {CRYPTO} 2000,
  20th Annual International Cryptology Conference, Santa Barbara, California,
  USA, August 20-24, 2000, Proceedings. Lecture Notes in Computer Science, vol.
  1880, pp. 197--215. Springer (2000)

\bibitem{DBLP:conf/eurocrypt/BlackR02}
Black, J., Rogaway, P.: A block-cipher mode of operation for parallelizable
  message authentication. In: Knudsen, L.R. (ed.) Advances in Cryptology -
  {EUROCRYPT} 2002, International Conference on the Theory and Applications of
  Cryptographic Techniques, Amsterdam, The Netherlands, April 28 - May 2, 2002,
  Proceedings. Lecture Notes in Computer Science, vol. 2332, pp. 384--397.
  Springer (2002)

\bibitem{boneh11}
Boneh, D., Dagdelen, {\"O}., Fischlin, M., Lehmann, A., Schaffner, C., Zhandry,
  M.: Random oracles in a quantum world. In: Lee, D., Wang, X. (eds.) Advances
  in Cryptology -- ASIACRYPT 2011, Lecture Notes in Computer Science, vol.
  7073, pp. 41--69. Springer Berlin Heidelberg (2011)

\bibitem{DBLP:conf/eurocrypt/BonehZ13}
Boneh, D., Zhandry, M.: Quantum-secure message authentication codes. In:
  Johansson, T., Nguyen, P.Q. (eds.) Advances in Cryptology - {EUROCRYPT} 2013,
  32nd Annual International Conference on the Theory and Applications of
  Cryptographic Techniques, Athens, Greece, May 26-30, 2013. Proceedings.
  Lecture Notes in Computer Science, vol. 7881, pp. 592--608. Springer (2013)

\bibitem{DBLP:conf/crypto/BonehZ13}
Boneh, D., Zhandry, M.: Secure signatures and chosen ciphertext security in a
  quantum computing world. In: Canetti, R., Garay, J.A. (eds.) Advances in
  Cryptology - {CRYPTO} 2013 - 33rd Annual Cryptology Conference, Santa
  Barbara, CA, USA, August 18-22, 2013. Proceedings, Part {II}. Lecture Notes
  in Computer Science, vol. 8043, pp. 361--379. Springer (2013)

\bibitem{brassard2011merkle}
Brassard, G., H{\o}yer, P., Kalach, K., Kaplan, M., Laplante, S., Salvail, L.:
  Merkle puzzles in a quantum world. In: Advances in Cryptology--CRYPTO 2011,
  pp. 391--410. Springer (2011)

\bibitem{broadbent2015quantum}
Broadbent, A., Jeffery, S.: Quantum homomorphic encryption for circuits of low
  {T}-gate complexity. In: Advances in Cryptology--CRYPTO 2015, pp. 609--629.
  Springer (2015)

\bibitem{DBLP:conf/stoc/CarterW77}
Carter, L., Wegman, M.N.: Universal classes of hash functions (extended
  abstract). In: Hopcroft, J.E., Friedman, E.P., Harrison, M.A. (eds.)
  Proceedings of the 9th Annual {ACM} Symposium on Theory of Computing, May
  4-6, 1977, Boulder, Colorado, {USA}. pp. 106--112. {ACM} (1977)

\bibitem{DBLP:conf/sacrypt/CoglianiMNCRVV14}
Cogliani, S., Maimut, D., Naccache, D., do~Canto, R.P., Reyhanitabar, R.,
  Vaudenay, S., Viz{\'{a}}r, D.: {OMD:} {A} compression function mode of
  operation for authenticated encryption. In: Joux, A., Youssef, A.M. (eds.)
  Selected Areas in Cryptography - {SAC} 2014 - 21st International Conference,
  Montreal, QC, Canada, August 14-15, 2014, Revised Selected Papers. Lecture
  Notes in Computer Science, vol. 8781, pp. 112--128. Springer (2014)

\bibitem{DBLP:journals/jmc/DaemenR07}
Daemen, J., Rijmen, V.: Probability distributions of correlation and
  differentials in block ciphers. J. Mathematical Cryptology  1(3),  221--242
  (2007)

\bibitem{DBLP:conf/icits/DamgardFNS13}
Damg{\aa}rd, I., Funder, J., Nielsen, J.B., Salvail, L.: Superposition attacks
  on cryptographic protocols. In: Padr{\'{o}}, C. (ed.) Information Theoretic
  Security - 7th International Conference, {ICITS} 2013, Singapore, November
  28-30, 2013, Proceedings. Lecture Notes in Computer Science, vol. 8317, pp.
  142--161. Springer (2013)

\bibitem{FIPS-800-38B}
Dworkin, M.: {Recommendation for Block Cipher Modes of Operation: The CMAC Mode
  for Authentication}. {NIST} Special Publication 800-38B, National Institute
  for Standards and Technology (May 2005)

\bibitem{DBLP:journals/joc/EvenM97}
Even, S., Mansour, Y.: A construction of a cipher from a single pseudorandom
  permutation. J. Cryptology  10(3),  151--162 (1997)

\bibitem{gagliardoni2015semantic}
Gagliardoni, T., H{\"u}lsing, A., Schaffner, C.: Semantic security and
  indistinguishability in the quantum world. arXiv preprint arXiv:1504.05255
  (2015)

\bibitem{DBLP:conf/stoc/Grover96}
Grover, L.K.: A fast quantum mechanical algorithm for database search. In:
  Miller, G.L. (ed.) Proceedings of the Twenty-Eighth Annual {ACM} Symposium on
  the Theory of Computing, Philadelphia, Pennsylvania, USA, May 22-24, 1996.
  pp. 212--219. {ACM} (1996)

\bibitem{DBLP:conf/eurocrypt/HoangKR15}
Hoang, V.T., Krovetz, T., Rogaway, P.: Robust authenticated-encryption {AEZ}
  and the problem that it solves. In: Oswald, E., Fischlin, M. (eds.) Advances
  in Cryptology - {EUROCRYPT} 2015 - 34th Annual International Conference on
  the Theory and Applications of Cryptographic Techniques, Sofia, Bulgaria,
  April 26-30, 2015, Proceedings, Part {I}. Lecture Notes in Computer Science,
  vol. 9056, pp. 15--44. Springer (2015)

\bibitem{DBLP:conf/fse/IwataK03}
Iwata, T., Kurosawa, K.: {OMAC:} one-key {CBC} {MAC}. In: Johansson, T. (ed.)
  Fast Software Encryption, 10th International Workshop, {FSE} 2003, Lund,
  Sweden, February 24-26, 2003, Revised Papers. Lecture Notes in Computer
  Science, vol. 2887, pp. 129--153. Springer (2003)

\bibitem{DBLP:conf/fse/IwataM0M14}
Iwata, T., Minematsu, K., Guo, J., Morioka, S.: {CLOC:} authenticated
  encryption for short input. In: Cid, C., Rechberger, C. (eds.) Fast Software
  Encryption - 21st International Workshop, {FSE} 2014, London, UK, March 3-5,
  2014. Revised Selected Papers. Lecture Notes in Computer Science, vol. 8540,
  pp. 149--167. Springer (2014)

\bibitem{DBLP:journals/corr/Kaplan14}
Kaplan, M.: Quantum attacks against iterated block ciphers. CoRR  abs/1410.1434
  (2014)

\bibitem{DBLP:journals/corr/KaplanLLN15}
Kaplan, M., Leurent, G., Leverrier, A., Naya{-}Plasencia, M.: Quantum
  differential and linear cryptanalysis. CoRR  abs/1510.05836 (2015)

\bibitem{DBLP:conf/fse/KrovetzR11}
Krovetz, T., Rogaway, P.: The software performance of authenticated-encryption
  modes. In: Joux, A. (ed.) Fast Software Encryption - 18th International
  Workshop, {FSE} 2011, Lyngby, Denmark, February 13-16, 2011, Revised Selected
  Papers. Lecture Notes in Computer Science, vol. 6733, pp. 306--327. Springer
  (2011)

\bibitem{5513654}
Kuwakado, H., Morii, M.: Quantum distinguisher between the 3-round {Feistel}
  cipher and the random permutation. In: Information Theory Proceedings (ISIT),
  2010 IEEE International Symposium on. pp. 2682--2685 (June 2010)

\bibitem{6400943}
Kuwakado, H., Morii, M.: Security on the quantum-type {Even-Mansour} cipher.
  In: Information Theory and its Applications (ISITA), 2012 International
  Symposium on. pp. 312--316 (Oct 2012)

\bibitem{DBLP:journals/joc/LiskovRW11}
Liskov, M., Rivest, R.L., Wagner, D.: Tweakable block ciphers. J. Cryptology
  24(3),  588--613 (2011)

\bibitem{DBLP:journals/siamcomp/LubyR88}
Luby, M., Rackoff, C.: How to construct pseudorandom permutations from
  pseudorandom functions. {SIAM} J. Comput.  17(2),  373--386 (1988)

\bibitem{lydersen2010hacking}
Lydersen, L., Wiechers, C., Wittmann, C., Elser, D., Skaar, J., Makarov, V.:
  Hacking commercial quantum cryptography systems by tailored bright
  illumination. Nature photonics  4(10),  686--689 (2010)

\bibitem{DBLP:conf/indocrypt/McGrewV04}
McGrew, D.A., Viega, J.: The security and performance of the galois/counter
  mode {(GCM)} of operation. In: Canteaut, A., Viswanathan, K. (eds.) Progress
  in Cryptology - {INDOCRYPT} 2004, 5th International Conference on Cryptology
  in India, Chennai, India, December 20-22, 2004, Proceedings. Lecture Notes in
  Computer Science, vol. 3348, pp. 343--355. Springer (2004)

\bibitem{DBLP:conf/eurocrypt/Minematsu14}
Minematsu, K.: Parallelizable rate-1 authenticated encryption from pseudorandom
  functions. In: Nguyen, P.Q., Oswald, E. (eds.) Advances in Cryptology -
  {EUROCRYPT} 2014 - 33rd Annual International Conference on the Theory and
  Applications of Cryptographic Techniques, Copenhagen, Denmark, May 11-15,
  2014. Proceedings. Lecture Notes in Computer Science, vol. 8441, pp.
  275--292. Springer (2014)

\bibitem{MW13}
Montanaro, A., de~Wolf, R.: A survey of quantum property testing. arXiv
  preprint arXiv:1310.2035  (2013)

\bibitem{DBLP:conf/asiacrypt/Rogaway04}
Rogaway, P.: Efficient instantiations of tweakable blockciphers and refinements
  to modes {OCB} and {PMAC}. In: Lee, P.J. (ed.) Advances in Cryptology -
  {ASIACRYPT} 2004, 10th International Conference on the Theory and Application
  of Cryptology and Information Security, Jeju Island, Korea, December 5-9,
  2004, Proceedings. Lecture Notes in Computer Science, vol. 3329, pp. 16--31.
  Springer (2004)

\bibitem{DBLP:conf/ccs/RogawayBBK01}
Rogaway, P., Bellare, M., Black, J., Krovetz, T.: {OCB:} a block-cipher mode of
  operation for efficient authenticated encryption. In: Reiter, M.K., Samarati,
  P. (eds.) {CCS} 2001, Proceedings of the 8th {ACM} Conference on Computer and
  Communications Security, Philadelphia, Pennsylvania, USA, November 6-8, 2001.
  pp. 196--205. {ACM} (2001)

\bibitem{SS}
Santoli, T., Schaffner, C.: Using simon's algorithm to attack symmetric-key
  cryptographic primitives. arXiv preprint arXiv:1603.07856  (2016)

\bibitem{CAESAR_Minalpher}
Sasaki, Y., Todo, Y., Aoki, K., Naito, Y., Sugawara, T., Murakami, Y., Matsui,
  M., Hirose, S.: {Minalpher v1.1}. CAESAR submission (August 2015)

\bibitem{DBLP:journals/siamcomp/Shor97}
Shor, P.W.: Polynomial-time algorithms for prime factorization and discrete
  logarithms on a quantum computer. {SIAM} J. Comput.  26(5),  1484--1509
  (1997)

\bibitem{simon1997power}
Simon, D.R.: On the power of quantum computation. SIAM journal on computing
  26(5),  1474--1483 (1997)

\bibitem{unruh15}
Unruh, D.: Non-interactive zero-knowledge proofs in the quantum random oracle
  model. In: Eurocrypt 2015. vol. 9057, pp. 755--784. Springer (2015), preprint
  on IACR ePrint 2014/587

\bibitem{xu2010experimental}
Xu, F., Qi, B., Lo, H.K.: Experimental demonstration of phase-remapping attack
  in a practical quantum key distribution system. New Journal of Physics
  12(11),  113026 (2010)

\bibitem{DBLP:conf/fse/Yuval97}
Yuval, G.: Reinventing the travois: Encryption/mac in 30 {ROM} bytes. In:
  Biham, E. (ed.) Fast Software Encryption, 4th International Workshop, {FSE}
  '97, Haifa, Israel, January 20-22, 1997, Proceedings. Lecture Notes in
  Computer Science, vol. 1267, pp. 205--209. Springer (1997)

\bibitem{DBLP:conf/focs/Zhandry12}
Zhandry, M.: How to construct quantum random functions. In: 53rd Annual {IEEE}
  Symposium on Foundations of Computer Science, {FOCS} 2012, New Brunswick, NJ,
  USA, October 20-23, 2012. pp. 679--687. {IEEE} Computer Society (2012)

\bibitem{zhandry2015secure}
Zhandry, M.: Secure identity-based encryption in the quantum random oracle
  model. International Journal of Quantum Information  13(04),  1550014 (2015)

\bibitem{zhao2008quantum}
Zhao, Y., Fung, C.H.F., Qi, B., Chen, C., Lo, H.K.: Quantum hacking:
  Experimental demonstration of time-shift attack against practical
  quantum-key-distribution systems. Physical Review A  78(4),  042333 (2008)

\end{thebibliography}
\end{document}